\documentclass[parindent=0pt]{article}
\usepackage{booktabs}
\usepackage[T1]{fontenc}
\usepackage[utf8]{inputenc}
\usepackage{makecell}
\usepackage{subcaption}

\usepackage[utf8]{inputenc}     
\usepackage[english]{babel}     
\usepackage[a4paper, left=1in, right=1in, top=1.25in, bottom=1.25in]{geometry}
\usepackage{graphicx}           

\usepackage{amsmath,amssymb}    
\usepackage{amsthm}             
\usepackage{mathtools}          
\usepackage{enumitem}           
\usepackage{listings}           
\usepackage{todonotes}          
\usepackage{hyperref}           
\usepackage{float}
\usepackage{amsthm}
\usepackage{amsmath}
\usepackage{amscd}
\usepackage{amssymb}
\usepackage{mathrsfs}
\usepackage{pdfpages}
\usepackage{setspace}
\usepackage{physics}
\usepackage{tcolorbox}

\usepackage{xcolor}
\usepackage{tikz}
\usepackage[toc,page,title,titletoc,header]{appendix}
 \usepackage{url}
 \usepackage{babel} 
\usepackage{csquotes}
\usepackage[f]{esvect}
\usepackage{blindtext}

\newcommand{\N}{\mathbb{N}}
\newcommand{\E}{\mathbb{E}}
\newcommand{\C}{\mathbb{C}}

\newcommand{\supp}{\mbox{supp}}

\newcommand{\g}{\mathfrak{g}}

\newcommand{\diag}{\mbox{diag}}

\newcommand{\oo}{\mathfrak{o}}

\let\phi=\varphi
\let\epsilon=\varepsilon

\newtheorem{theorem}{Theorem}
\newtheorem{lemma}[theorem]{Lemma}
\newtheorem{result}[theorem]{Result}
\newtheorem{definition}[theorem]{Definition}

\usepackage{authblk} 

\newcommand{\adots}{\mathinner{\mkern1mu \raise1pt
\hbox{.} \mkern2mu \raise4pt \hbox{.} \mkern2mu \raise7pt
\vbox{\kern7pt \hbox{.}} \mkern1mu}}
\def\[{[\![}
\def\]{]\!]}

\newlist{regimes}{enumerate}{1}
\setlist[regimes]{label={\textbullet\ Regime (\arabic*):}, leftmargin=*}

\title{%
Spectral properties of deterministic matrices multiplied by rotationally invariant random non-Hermitian ensembles
}

  \author[1,2]{Pierre Bousseyroux\thanks{Email: pierre.bousseyroux@polytechnique.edu}}
\author[3]{Marc Potters}

\affil[1]{Econophysics Lab, Institut Louis Bachelier, 28 Pl. de la Bourse, Palais Brongniart, 75002 Paris, France}
\affil[2]{LadHyX, UMR CNRS 7646, Ecole Polytechnique, Institut Polytechnique de Paris, 91128 Palaiseau, France}
\affil[3]{Capital Fund Management, Paris, France}

\newenvironment{remarks}{
  \par\vspace{1ex}
  \noindent\textbf{Remarks.}\begin{itemize}\setlength\itemsep{0pt}}
  {\end{itemize}\par\vspace{1ex}}

\makeatletter
\renewcommand{\@fnsymbol}[1]{%
  \ifcase#1\or
    \ensuremath{\dagger}\or
    \ensuremath{\ddagger}\or
    \ensuremath{\mathsection}\or
    \ensuremath{\mathparagraph}\or
    \ensuremath{\|}\or
    \ensuremath{\dagger\dagger}\or
    \ensuremath{\ddagger\ddagger}\or
    \ensuremath{\mathsection\mathsection}\or
    \ensuremath{\mathparagraph\mathparagraph}\else
    \@ctrerr
  \fi}
\makeatother

\begin{document}

\maketitle

\begin{abstract}
In this paper, we study spectral properties of multiplicative deformations of non-Hermitian random matrices. We consider matrices of the form $\vb{A}\vb{B}$, where $\vb{A}$ is a deterministic $N\times N$ matrix (not necessarily Hermitian) and $\vb{B}$ is a rotationally invariant random matrix. We show that, as $N\to\infty$, the boundary of the complex eigenvalue distribution of $\vb{A}\vb{B}$ is governed by simple equations involving the $\mathcal{R}_1$ and $\mathcal{R}_2$ transforms of $\vb{B}$ introduced in~\cite{bousseyroux1}.
\end{abstract}

\section{Introduction}

Random matrix theory often reveals that spectral problems with highly
non-explicit finite-dimensional expressions admit simple deterministic
descriptions in the large-$N$ limit. In the Hermitian setting, this principle is
formalized by Voiculescu's free probability theory: under suitable asymptotic
freeness assumptions, the limiting eigenvalue distribution of sums and products
of Hermitian matrices is governed by free additive and multiplicative
convolutions, encoded respectively by the $R$- and $S$-transforms
\cite{Voiculescu1986addition,Voiculescu1987multiplication,NicaSpeicher2006}.

For non-Hermitian matrices, the situation is considerably more involved. 
The additive problem has been studied using diagrammatic and hermitization methods
\cite{JanikEtAlPRE1997,janik1997non,FeinbergZeeNPB501,FeinbergZeeNPB504,
ChalkerWangPRL1997,Rodgers2010}, and more recently through the $\mathcal{R}$-transform
formalism developed in \cite{bousseyroux1}. In this setting, the classical
$R$-transform is replaced by suitable non-Hermitian generalizations. 

In the present paper, we address the multiplicative problem $\vb A\vb B$.
Compared with the additive case, this problem is much less developed in the
non-Hermitian setting. Existing results mainly rely on diagrammatic methods,
either for products of independent Gaussian matrices
\cite{BurdaJanikNowak2011}, or through the introduction of a non-Hermitian
analogue of the $S$-transform \cite{BurdaJanikNowak2011}. Another possible
approach is operator-valued free probability, where operator-valued
$S$-transforms exist and satisfy twisted multiplicativity properties
\cite{Dykema2006,Speicher2023TwistedS}. However, these works also show that the
resulting relations are substantially more complicated than in the Hermitian
case: they do not seem to reduce to an
effective relation between non-Hermitian $R$- and $S$-type transforms, and
therefore do not provide a practical formula for studying $\vb A\vb B$.

A particularly simple and important case is the bi-invariant one.\footnote{We
say that a random matrix $\vb M$ is bi-invariant if $\vb M$ and
$\vb U\vb M\vb V$ have the same distribution for all unitary matrices
$\vb U,\vb V$.} This is one of the few non-Hermitian settings where a
well-developed and explicit theory is available. In this case, the limiting
eigenvalue distribution is rotationally invariant and can be related to the
limiting singular-value distribution. This is the content of the single ring
picture, first obtained in the physics literature using diagrammatic methods
\cite{feinberg1997non}, then understood in the framework of Brown measures of
$R$-diagonal elements \cite{haagerup2000brown}, and finally proved rigorously
for large bi-invariant random matrices by Guionnet, Krishnapur and Zeitouni
\cite{guionnet2011single}. More precisely, if $\vb{A}$ is bi-invariant and if
\begin{equation}\label{defF}
    F_{\vb A}(r)
    =
    \rho_{\vb A}\big(\{z\in\mathbb C:\ |z|\le r\}\big)
\end{equation}
denotes the radial cumulative distribution function of its limiting eigenvalue
distribution $\rho_{\vb A}$, then the Haagerup--Larsen formula gives
\begin{equation}\label{singlering}
    S_{\vb A\vb A^*}\big(F_{\vb A}(r)-1\big)
    =
    \frac{1}{r^2},
\end{equation}
where $S_{\vb A\vb A^*}$ denotes the scalar $S$-transform of the limiting
eigenvalue distribution of $\vb A\vb A^*$. 

This also makes products of free bi-invariant matrices particularly simple. If
$\vb{A}$ and $\vb{B}$ are free bi-invariant matrices, we have

\begin{equation}\label{multiplication_bi}
        S_{\vb A\vb A^*}\big(F_{\vb A\vb B}(r)-1\big)S_{\vb B\vb B^*}\big(F_{\vb A\vb B}(r)-1\big)
    =
    \frac{1}{r^2}
\end{equation} and so, the radial quantiles multiply:
\begin{equation}
    F_{\vb A\vb B}^{-1}(p)
    =
    F_{\vb A}^{-1}(p)\,
    F_{\vb B}^{-1}(p).
    \label{eq:bi-invariant-quantiles}
\end{equation}
A consequence is that the limiting eigenvalue density of a product of independent
identically distributed matrices coincides with that of a power of a single
matrix drawn from the same ensemble \cite{BurdaNowakSwiech2012}.

The setting of the present paper is more general. We consider multiplicative deformations $\vb{A}\vb{B}$, where $\vb A$ is a deterministic $N\times N$ matrix, not necessarily Hermitian, and $\vb B$ is a rotationally invariant random matrix, in the sense that $\vb{U}\vb{B}\vb{U}^*$ has the same distribution as $\vb{B}$ for any unitary matrix $\vb{U}$. We assume that the random matrices considered in this
work admit a continuous limiting spectral distribution, denoted by $\rho$, in
the complex plane. This distribution is defined as the large-$N$ limit of the
empirical spectral distribution of the complex eigenvalues. From a
mathematical point of view, the natural limiting object is the Brown measure
\cite{Brown1986Lidskii,HaagerupSchultz2007}. We denote by
$\supp(\rho)$ the support of $\rho$, and by $\partial\supp(\rho)$ its boundary,
which corresponds to the spectral edges.

We will also consider the eigenvector self-overlap $\mathcal O(z)$, introduced
by Janik et al.~\cite{janik1999correlations}. This quantity measures the
non-orthogonality of eigenvectors and plays an important role in non-normal
random matrix theory. Throughout the paper we use $\mathcal O(z)=\mathfrak o(z)^2$. 

In this paper, we derive both the spectral density, or equivalently the associated Stieltjes transform, and the self-overlap of $\vb{A}\vb{B}$ from the transforms $\mathcal R_{1, \vb{B}}$ and $\mathcal R_{2, \vb{B}}$ introduced in \cite{bousseyroux1}, together with deterministic quantities associated with $\vb A$. A first application is a relatively simple description of the spectral boundaries of $\vb A\vb B$. These boundary equations can be written in terms of scalar functions which play, for the edge problem, a role analogous to the Hermitian $S$-transform, although they are not sufficient to determine the full two-dimensional spectrum. The method follows the non-Hermitian transform formalism developed in \cite{bousseyroux1}. Beyond their theoretical interest, spectral boundaries play a key role in stability analysis, since the location of spectral edges often signals transitions such as the onset of chaos or the loss of stability.

\section{Main result}

Before stating the main result, we give a brief recapitulation of several transforms that play an important role in the Hermitian and non-Hermitian settings. This is mainly an opportunity to introduce the tools that will be needed in the sequel.

A standard way to approach the eigenvalue distribution of a random matrix $\vb{M}$ is to introduce, in the large-$N$ limit, the Stieltjes transform of $\vb{M}$:
\begin{equation}
     \mathfrak{g}_{\vb{M}}(z)
     :=
     \tau\left[(z\vb{1}-\vb{M})^{-1}\right],
\end{equation}
where throughout the paper we denote
\begin{equation}
    \tau := \lim_{N\to+\infty}\frac{1}{N}\tr.
\end{equation}

The knowledge of $\mathfrak{g}_{\vb{M}}(z)$ allows one to recover the limiting spectral distribution $\rho_{\vb{M}}$ of $\vb{M}$. In the one-dimensional real case, when $\vb{M}$ is Hermitian, one has
\begin{equation}
    \lim_{\epsilon\downarrow 0}\Im \mathfrak{g}_{\vb{M}}(x-i\epsilon)
    =
    \pi \rho_{\vb{M}}(x).
\end{equation}
In the non-Hermitian case, the density is recovered through the two-dimensional analogue, namely Gauss's law,
\begin{equation}
    \partial_{\bar z}\mathfrak{g}_{\vb{M}}(z)
    =
    \pi \rho_{\vb{M}}(z).
\end{equation}

To study $\vb{A}+\vb{B}$ in the Hermitian setting, one can show that there exists a function $R_{\vb{B}}$, called the $R$-transform of $\vb{B}$, such that, for $z$ sufficiently large,
\begin{equation}\label{defR}
    \mathfrak{g}_{\vb{A}+\vb{B}}(z)
    =
    \mathfrak{g}_{\vb{A}}
    \left(
        z - R_{\vb{B}}\left(\mathfrak{g}_{\vb{A}+\vb{B}}(z)\right)
    \right).
\end{equation}
See, for example, the classical works on free probability~\cite{voiculescu1992free,NicaSpeicher2006}, and also~\cite{bun2016rotational} for a derivation using replicas. One obtains as a corollary that, if $\vb{A}$ and $\vb{B}$ are free, then the $R$-transform is additive. From~\eqref{defR}, by taking $\vb{A}=0$, one may equivalently define $R$ by
\begin{equation}
    R(g) = z(g) - \frac{1}{g},
\end{equation}
where $z(g)$ denotes the functional inverse of the Stieltjes transform.

We would also like to have an analogous subordination relation in the multiplicative case. For multiplicative questions, it is more convenient to use the $\mathfrak{t}$-transform, defined by
\begin{equation}
    \mathfrak{t}_{\vb{M}}(z)
    :=
    \tau\left[
        \vb{M}\left(z\vb{1}-\vb{M}\right)^{-1}
    \right]
    =
    z\mathfrak{g}_{\vb{M}}(z)-1.
\end{equation}
One can show, see for instance~\cite{voiculescu1992free,potters2020first}, and also~\cite{bun2016rotational} for the replica approach, that there exists a function $S_{\vb{B}}$, called the $S$-transform of $\vb{B}$, such that, if
\begin{equation}
    \vb{M} = \sqrt{\vb{B}}\vb{A}\sqrt{\vb{B}},
\end{equation}
then
\begin{equation}\label{defS}
    \mathfrak{t}_{\vb{M}}(z)
    =
    \mathfrak{t}_{\vb{A}}
    \left(
        z S_{\vb{B}}\left(\mathfrak{t}_{\vb{M}}(z)\right)
    \right).
\end{equation}
In particular, if $\vb{A}$ and $\vb{B}$ are two Hermitian random matrices, then the result, first obtained in~\cite{voiculescu1992free}, reads
\begin{equation}
    S_{\vb{M}} = S_{\vb{A}}S_{\vb{B}}.
\end{equation}
From~\eqref{defS}, by taking $\vb{A}=\vb{1}$, one may define the $S$-transform by
\begin{equation}\label{definitionS_classical}
    S_{\vb{M}}(t)
    :=
    \frac{t+1}{t\,\mathfrak{t}_{\vb{M}}^{-1}(t)},
\end{equation}
where $\mathfrak{t}_{\vb{M}}^{-1}$ is the functional inverse of the $\mathfrak{t}$-transform.

The $R$- and $S$-transforms are related, see for instance~\cite{potters2020first}:
\begin{equation}\label{relationRS}
\begin{aligned}
    R(g) &= \frac{1}{S(R(g)g)},
    &
    S(g) &= \frac{1}{R(gS(g))}.
\end{aligned}
\end{equation}

This is the reasoning in the Hermitian multiplicative setting that we would like to extend to the non-Hermitian setting. We strongly encourage the reader to consult~\cite{bousseyroux1}, where the notions we will use here are developed in detail. Let $\vb{M}$ be a large random matrix. We will use the following notation:
\begin{equation}
    \mathfrak{g}_{1,\vb{M}}(\omega,z)
    =
    \tau\left[
        \omega
        \left(
            \omega^2\vb{1}
            -
            (z\vb{1}-\vb{M})(z\vb{1}-\vb{M})^*
        \right)^{-1}
    \right],
\end{equation}
and
\begin{equation}
    \mathfrak{g}_{2,\vb{M}}(\omega,z)
    =
    -
    \tau\left[
        (z\vb{1}-\vb{M})^*
        \left(
            \omega^2\vb{1}
            -
            (z\vb{1}-\vb{M})(z\vb{1}-\vb{M})^*
        \right)^{-1}
    \right],
\end{equation}
as well as
\begin{equation}
    \mathcal{G}_{\vb{M}}(\omega,z)
    =
    \begin{pmatrix}
        \mathfrak{g}_{1,\vb{M}}(\omega,z)
        &
        \mathfrak{g}_{2,\vb{M}}(\omega,z)
        \\
        \overline{\mathfrak{g}_{2,\vb{M}}(\omega,z)}
        &
        \mathfrak{g}_{1,\vb{M}}(\omega,z)
    \end{pmatrix}.
\end{equation}

The analogue of~\eqref{defR} is the main result of~\cite{bousseyroux1}, which we restate here. In that article, the functions $\mathcal{R}_1$ and $\mathcal{R}_2$ are introduced as functions derived from a certain $\mathcal{H}$-transform.

\begin{result}\label{main_result}
Let $\vb{A}$ be a large deterministic matrix and let $\vb{B}$ be a rotationally invariant random matrix. Then, for large $\omega$ and $z\in\C$, one expects
\begin{equation}\label{eq_result}
    \mathcal{G}_{\vb{A}+\vb{B}}(\omega,z)
    =
    \mathcal{G}_{\vb{A}}
    \left(
        \omega-\mathcal{R}_{1,\vb{B}}(\mathfrak{g}_1,\mathfrak{g}_2),
        z-\mathcal{R}_{2,\vb{B}}(\mathfrak{g}_1,\mathfrak{g}_2)
    \right),
\end{equation}
where
\begin{equation}
    \mathfrak{g}_1
    =
    \mathfrak{g}_{1,\vb{A}+\vb{B}}(\omega,z)
\end{equation}
and
\begin{equation}
    \mathfrak{g}_2
    =
    \mathfrak{g}_{2,\vb{A}+\vb{B}}(\omega,z).
\end{equation}
This identity can be analytically continued to all $(\omega,z)$, although one must carefully choose the appropriate branches of $\mathcal{R}_1$ and $\mathcal{R}_2$.
\end{result}

Stating this result has allowed us to introduce the functions $\mathcal{R}_1$ and $\mathcal{R}_2$, and, more importantly, to emphasize the role played by the choice of branches. We would like to find an analogue of this result in the multiplicative setting, namely a relation expressing $\mathcal{G}_{\vb{A}\vb{B}}$ implicitly in terms of $\mathcal{G}_{\vb{A}}$. This appears to be a difficult problem. In particular, finding a subordination relation similar to~\eqref{defS} seems more involved in the non-Hermitian setting. Instead, we prove the following main result.

\begin{result}\label{mainresult}
Let $\vb{A}$ be a deterministic or random matrix, and let $\vb{B}$ be a rotationally invariant random matrix, independent of $\vb{A}$ whenever $\vb{A}$ is random.  Let $z\in\C$. The following procedure gives $\mathfrak{g}(z):=\mathfrak{g}_{\vb{A}\vb{B}}(z)$ and $\mathfrak{o}(z):=\mathfrak{o}_{\vb{A}\vb{B}}(z)$.

Find $G$ and $E$ such that
\begin{equation}\label{system_zero}
    \left\{
    \begin{aligned}
        G
        &=
        \tau\left[
            -\mathcal{R}_{1, \vb{B}}(G,E)\vb{A}\vb{A}^*
            \left(
                \mathcal{R}_{1, \vb{B}}(G,E)^2\vb{A}\vb{A}^*
                -
                \left|z\vb{1}-\mathcal{R}_{2, \vb{B}}(G,E)\vb{A}\right|^2
            \right)^{-1}
        \right],
        \\[0.6em]
        E
        &=
        -\tau\left[
            \vb{A}
            \left(z\vb{1}-\mathcal{R}_{2, \vb{B}}(G,E)\vb{A}\right)^*
            \left(
                \mathcal{R}_{1, \vb{B}}(G,E)^2\vb{A}\vb{A}^*
                -
                \left|z\vb{1}-\mathcal{R}_{2, \vb{B}}(G,E)\vb{A}\right|^2
            \right)^{-1}
        \right],
    \end{aligned}
    \right.
\end{equation}
where
\begin{equation}
    |X|^2 := XX^*.
\end{equation}
Then
\begin{equation}
    \mathfrak{g}(z)
    =
    -\tau\left[
        \left(z\vb{1}-\mathcal{R}_{2, \vb{B}}(G,E)\vb{A}\right)^*
        \left(
            \mathcal{R}_{1, \vb{B}}(G,E)^2\vb{A}\vb{A}^*
            -
            \left|z\vb{1}-\mathcal{R}_{2, \vb{B}}(G,E)\vb{A}\right|^2
        \right)^{-1}
    \right],
\end{equation}
and
\begin{equation}
    \mathfrak{o}(z)^2
    =
    G\mathcal{R}_{1, \vb{B}}(G,E)
    \,
    \tau\left[
        \left(
            \mathcal{R}_{1, \vb{B}}(G,E)^2\vb{A}\vb{A}^*
            -
            \left|z\vb{1}-\mathcal{R}_{2, \vb{B}}(G,E)\vb{A}\right|^2
        \right)^{-1}
    \right].
\end{equation}
Here $\mathcal{R}_1$ and $\mathcal{R}_2$ denote appropriate branches.
\end{result}

\begin{remarks}
    \item The case $\vb{B}\vb{A}$ can be treated by taking the conjugate.

    \item In this formula, the transforms $\mathfrak{g}_{1,\vb{A}}$ and $\mathfrak{g}_{2,\vb{A}}$ do not appear in a simple way. This probably explains why finding an $S$-transform generalizing~\eqref{defS} in the non-Hermitian case is difficult.

    \item When $\vb{A}=\vb{1}$, one recovers the known relations presented in~\cite{bousseyroux1}.

    \item Note that taking the limit \(\omega \to 0\) of the diagonal coefficients in Eq.~\eqref{eq_result}, when \(z\) lies outside the spectrum of \(\vb{A}+\vb{B}\) and in a region where \(\g_{\vb{A}+\vb{B}}(z)\) is not identically zero, so that \(\mathcal{R}_1\) does not diverge, yields
\begin{equation}
    \mathfrak{g}_{\vb{A}+\vb{B}}(z)
    =
    \mathfrak{g}_{\vb{A}}
    \left(
        z-\mathcal{R}_{2,\vb{B}}(0,\mathfrak{g}_{\vb{A}+\vb{B}}(z))
    \right).
\end{equation}
This argument has already been carried out in Result~4 of~\cite{bousseyroux1}. The above equation is the analogue of~\eqref{defR} in the non-Hermitian setting, where \(\mathcal{R}_{2,\vb{B}}\) denotes an appropriate branch. Note that, when $|z|$ is large, then $\mathcal{R}_{2,\vb{B}}$ is the principal branch.

Furthermore, Result~\ref{mainresult} gives the following system:
\begin{equation}\label{eq3}
    \mathfrak{t}_{\vb{A}\vb{B}}(z)
    =
    \mathfrak{t}_{\vb{A}}\left(
        \frac{z}{\mathcal{R}_{2,\vb{B}}(0,E)}
    \right),
\end{equation}
where \(E\) satisfies
\begin{equation}
    E\mathcal{R}_{2,\vb{B}}(0,E)
    =
    \mathfrak{t}_{\vb{A}\vb{B}}(z),
\end{equation}
and where \(z\) lies in a region where \(\g_{\vb{A}\vb{B}}(z)\) is not identically zero. In the next section, we shall rewrite equation~\eqref{eq3} so as to introduce a new object, namely the \(S\)-transform of \(\vb{B}\). Note that, when $\tau(\vb{A})=0$, $E=0$ is always a solution of the system. We then obtain
\begin{equation}
    \g_{\vb{A}\vb{B}}(z)=0.
\end{equation}
  \end{remarks}

\section{Applications}
\subsection{Case where $\vb{A}\vb{B}$ behaves as a bi-invariant matrix}

If $\vb{B}$ is a bi-invariant matrix, we know that $\vb{A}\vb{B}$ behaves as a bi-invariant matrix, and one can then use relation~\eqref{multiplication_bi}. One can generalize this fact. Let us apply the main Result~\ref{mainresult} to the case where $\tau(\vb{B}) = 0$ and where $\vb{A}$ has the same distribution as $-\vb{A}$. In this case, $E = 0$ is always a solution of the system~\ref{system_zero}. We then see that $\g_{\vb{A}\vb{B}}(z)$ and $\mathcal{O}_{\vb{A}\vb{B}}(z)$ only depend on $\vb{A}\vb{A}^*$, and therefore on its singular values. Thus, one can replace $\vb{A}$ by a bi-invariant matrix sharing the same singular values, and then apply formula~\eqref{multiplication_bi}.

\begin{result}\label{propcentre}
Let $\vb{A}$ be a random matrix such that $\vb{A}$ and $-\vb{A}$ follow the same distribution. Let $\vb{B}$ be a large rotationally invariant random matrix, independent of $\vb{A}$ and centered in the sense that $\tau(\vb{B}) = 0$. The spectrum of $\vb{A}\vb{B}$ is a ring and, for all $z\in \C$, one has
\begin{equation}\label{formula_produit}
    S_{\vb{A}\vb{A}^*}(F_{\vb{A}\vb{B}}(|z|)-1)
    S_{\vb{B}\vb{B}^*}(F_{\vb{A}\vb{B}}(|z|)-1)
    =
    \frac{1}{|z|^2}.
\end{equation}

Moreover,
\begin{equation}\label{formulaoverlap}
    \mathcal{O}_{\vb{A}\vb{B}}(z)
    =
    \frac{F_{\vb{A}\vb{B}}(|z|)(1-F_{\vb{A}\vb{B}}(|z|))}{|z|^2}.
\end{equation}
\end{result}

\begin{remarks}
\item The overlap formula~\eqref{formulaoverlap} follows directly from the overlap formula in the bi-invariant case given in~\cite{belinschi2017squared}.
\item The outer boundary is a circle of radius
$\sqrt{\tau(\vb{A}\vb{A}^*)\tau(\vb{B}\vb{B}^*)}$,
and the inner radius is given by
\begin{equation}
\frac{1}{\sqrt{\tau\left((\vb{A}\vb{A}^*)^{-1}\right)
\tau\left((\vb{B}\vb{B}^*)^{-1}\right)}}.
\end{equation}
\item We show in Appendix~\ref{from} how to recover equation~\eqref{formula_produit} using only Main Result~\ref{mainresult}. We then recover equation~\eqref{singlering} of the Single Ring Theorem.
\end{remarks}

\subsection{Extension of the scalar Hermitian $S$-transform}

\begin{definition}
Let $\vb{A}$ be a rotationally invariant matrix. We define
\begin{equation}
k^*(\vb{A})
=
\min \left\{
k \in \N^*,
\tau\!\left(\vb{A}^k\right) \neq 0
\right\},
\end{equation}
with values in $\N^* \cup \{+\infty\}$.
\end{definition}

\begin{remarks}
    \item For a bi-invariant matrix $\vb{A}$, one always has $k^*(\vb{A}) = +\infty$.
    \item For a Hermitian matrix $\vb{H}$, one has $k^*(\vb{H}) = +\infty$ if $\vb{H}$ is the zero matrix, $k^*(\vb{H}) = 2$ if $\vb{H}$ is centered and non-zero, and $k^*(\vb{H}) = 1$ otherwise.
\end{remarks}

The first thing one can do is to extend the notion of the scalar $S$-transform defined in \eqref{defS} to the non-Hermitian setting.

\begin{definition}\label{definitionS}
Let $\vb{B}$ be a non-Hermitian rotationally invariant matrix. For each branch of $\mathcal{R}_{2,\vb{B}}$, we define the associated $S$-transforms as the solutions of the equation
\begin{equation}\label{defSS}
    S_{\vb{B}}(g)
    =
    \frac{1}{
        \mathcal{R}_{2,\vb{B}}\left(0,gS_{\vb{B}}(g)\right)
    }.
\end{equation}

By convention, if one can find a branch of $\mathcal{R}_{2,\vb{B}}$ such that
\begin{equation}
    \mathcal{R}_{2,\vb{B}}(0,g)=0,
\end{equation}
then we set
\begin{equation}\label{Sbiinvariant}
    S_{\vb{B}}(t)=+\infty,
\end{equation}
and we say that $+\infty$ is an $S$-transform of $\vb{B}$. The $S$-transforms associated with the principal branch of $\mathcal{R}_{2,\vb{B}}$ will be called the principal $S$-transforms.

Furthermore, we say that $0$ is also an $S$-transform of $\vb{B}$ if one can find a function $\beta(\alpha)$, tending to $0$ as $\alpha\to 0$, such that
\begin{equation}\label{condition}
    \left|\mathcal{R}_{1,\vb{B}}(\alpha,\beta(\alpha))\right|
    \underset{\alpha\to 0}{\longrightarrow}
    +\infty.
\end{equation}
\end{definition}

\begin{remarks}
    \item The classical Hermitian $S$-transform is usually defined under the condition $\tau(\vb{A})\neq 0$. It was also extended in~\cite{rao2007multiplication} to the centered Hermitian case. Here, the definition may still make sense even when $\tau(\vb{A})=0$ and $\vb{A}$ is non-Hermitian.

    \item If $\vb{A}$ is a matrix such that $k:=k^*(\vb{A})\in \N^*$, then
    \begin{equation}
        \mathcal{R}_{2,\vb{A}}(0,g)
        \underset{g\to 0}{\sim}
        \tau\!\left(\vb{A}^{k}\right) g^{k-1}
    \end{equation}
    for the principal branch. Therefore, one should associate with $\vb{A}$ exactly $k$ different principal $S$-transforms, corresponding to the $k$ possible choices of the $k$-th root, and satisfying
    \begin{equation}\label{S0}
        S_{\vb{A}}(g)
        \underset{g\to 0}{\sim}
        \frac{1}{
            \sqrt[k]{\tau\!\left(\vb{A}^{k}\right) g^{k-1}}
        }.
    \end{equation}

    \item There may be several $S$-transforms if there are several branches of $\mathcal{R}_{2,\vb{B}}$.

    \item By the convention~\eqref{Sbiinvariant}, it follows that $+\infty$ is an $S$-transform for any bi-invariant matrix.
    
    \item These different $S$-transforms should not be regarded as analytic continuations of one another. Knowing one branch of the $S$-transform does not allow one, in general, to recover the others.

    \item As explained in Result~4 of \cite{bousseyroux1}, condition~\eqref{condition} can occur when the distribution of $\vb{B}$ gives rise to a charge-free region with vanishing field. A typical example is a bi-invariant matrix whose spectral support is a ring rather than a disk: the origin then lies inside the hole, where the associated field is zero.
\end{remarks}

Using this definition, one can then rewrite \eqref{eq3} as the following result.

\begin{result}
Let $\vb{A}$ and $\vb{B}$ be two large independent rotationally invariant random matrices. Let $z$ be outside the spectrum of $\vb{A}\vb{B}$. Then,
\begin{equation}\label{re3}
\mathfrak{t}_{\vb{A}\vb{B}}(z) = \mathfrak{t}_{\vb{A}}(z S_{\vb{B}}(\mathfrak{t}_{\vb{A}\vb{B}}(z))),
\end{equation}
where $S_{\vb{B}}$ denotes one $S$-transform of $\vb{B}$.

Note that, when $|z|$ is large, $S_{\vb{B}}$ is a principal $S$-transform of $\vb{B}$.
\end{result}

\begin{remarks}
\item Taking $\vb{A}=\vb{1}$, we obtain
\begin{equation}
S_{\vb{M}}(\mathfrak{t}_{\vb{M}}(z))
=
\frac{\mathfrak{t}_{\vb{M}}(z)+1}{\mathfrak{t}_{\vb{M}}(z)z},
\end{equation}
which is the analogue of \eqref{defS} in the Hermitian setting.

\item Let us recall that there are three cases in the definition of the $S$-transform of $\vb{B}$; see Def.~\ref{definitionS}. The first case is $S_{\vb{B}}=+\infty$, which occurs when $\vb{B}$ is bi-invariant. In this case, the right-hand side of equation~\eqref{re3} is equal to $0$. Hence $\mathfrak{t}_{\vb{A}\vb{B}}(z)=0$, meaning that $\vb{A}\vb{B}$ behaves like the zero matrix outside the spectrum. This was already known, since $\vb{A}\vb{B}$ is bi-invariant whenever $\vb{B}$ is bi-invariant.

The second case is $S_{\vb{B}}=0$. In this case, one obtains $\mathfrak{t}_{\vb{A}\vb{B}}(z)=\mathfrak{t}_{\vb{A}}(0)=-1$, and therefore $\mathfrak{g}_{\vb{A}\vb{B}}(z)=0$, which means that $z$ lies in a region where the field vanishes.

In all other cases, the formula
\begin{equation}\label{definitionS_nonhermitian_alt}
S_{\vb{M}}(t)
:=
\frac{t+1}{t \mathfrak{t}_{\vb{M}}^{-1}(t)},
\end{equation}
where $\mathfrak{t}_{\vb{M}}^{-1}$ denotes the functional inverse of the $\mathfrak{t}$-transform, provides some $S$-transforms of $\vb{M}$, but a priori not all of them.

\end{remarks}

A direct consequence of the previous result is the following result.

\begin{result}
Let $\vb{A}$ and $\vb{B}$ be two large independent rotationally invariant random matrices, and let
$\vb{M}=\vb{A}\vb{B}$. Let $S_{\vb{A}\vb{B}}$ be an $S$-transform of $\vb{A}\vb{B}$ such that, locally,
\begin{equation}\label{e}
    S_{\vb{A}\vb{B}}(\mathfrak{t}_{\vb{A}\vb{B}}(z))
    =
    \frac{\mathfrak{t}_{\vb{A}\vb{B}}(z)+1}
    {\mathfrak{t}_{\vb{A}\vb{B}}(z)z}.
\end{equation}
Then one can find an $S$-transform $S_{\vb{A}}$ of $\vb{A}$ and an $S$-transform $S_{\vb{B}}$ of $\vb{B}$ such that
\begin{equation}
    S_{\vb{A}\vb{B}}
    =
    S_{\vb{A}}S_{\vb{B}}.
\end{equation}

In particular, if $S_{\vb{A}\vb{B}}$ is a principal $S$-transform of $\vb{A}\vb{B}$, then one can find a principal $S$-transform $S_{\vb{A}}$ of $\vb{A}$ and a principal $S$-transform $S_{\vb{B}}$ of $\vb{B}$ such that
\begin{equation}
    S_{\vb{A}\vb{B}}
    =
    S_{\vb{A}}S_{\vb{B}}.
\end{equation}
\end{result}

\begin{remarks}
    \item If $0$ is neither an $S$-transform of $\vb{A}$ nor an $S$-transform of $\vb{B}$, then $0$ cannot be an $S$-transform of $\vb{A}\vb{B}$. Hence, it is not possible to have an empty charge region where the field vanishes.
    \item Consider a principal $S$-transform $S_{\vb{A}\vb{B}}$ such that Eq.~\eqref{e} holds for large $z$. Assume that $k^*(\vb{A})$ and $k^*(\vb{B})$ are finite integers greater than or equal to $2$. Suppose moreover that
\begin{equation}
    S_{\vb{A}\vb{B}}\neq +\infty.
\end{equation}
Then, using the previous result and Eq.~\eqref{S0}, there exists $C>0$ such that
\begin{equation}
    |S_{\vb{A}\vb{B}}(g)|
    \underset{g\to 0}{\sim}
    \frac{C}{
        |g|^{2-\frac{1}{k^*(\vb{A})}-\frac{1}{k^*(\vb{B})}}
    }.
\end{equation}
On the other hand, since $S_{\vb{A}\vb{B}}$ is assumed to be a principal $S$-transform, Eq.~\eqref{S0} gives
\begin{equation}
    |S_{\vb{A}\vb{B}}(g)|
    \underset{g\to 0}{\sim}
    \frac{C'}{
        |g|^{1-\frac{1}{k^*(\vb{A}\vb{B})}}
    },
\end{equation}
for some $C'>0$. Therefore,
\begin{equation}
    1-\frac{1}{k^*(\vb{A}\vb{B})}
    =
    2-\frac{1}{k^*(\vb{A})}
      -\frac{1}{k^*(\vb{B})}.
\end{equation}
This is impossible, since the left-hand side is strictly smaller than $1$, whereas the right-hand side is greater than or equal to $1$.

Thus, necessarily,
\begin{equation}
    S_{\vb{A}\vb{B}}=+\infty.
\end{equation}
We obtain the following consequence: if $\vb{A}$ and $\vb{B}$ are centered matrices, then $\vb{A}\vb{B}$ behaves like the zero matrix outside the spectrum, in the sense that $+\infty$ is the only principal $S$-transform of $\vb{A}\vb{B}$.
\end{remarks}

\subsection{Spectral boundaries}

For any matrix $\vb{A}$, we define
\begin{equation}
    f_{\vb{A}}(z)
    =
    \tau\left(\vb{A}\vb{A}^*|z-\vb{A}|^{-2}\right),
\end{equation}
where $z$ is a complex number. Following the definition of~\cite{bousseyroux3}, we also define
\begin{equation}
    h_{\vb{A}}(z)
    =
    \tau\left(|z-\vb{A}|^{-2}\right),
\end{equation}
where again $z$ is a complex number.

\begin{result}\label{result_boundaries}
Let $\vb{A}$ be a large deterministic matrix and let $\vb{B}$ be a rotationally invariant random matrix.

\begin{itemize}
    \item \textbf{Boundary of type 1:} If $\vb{B}$ is bi-invariant, then the spectral boundary of type 1 of $\vb{A}\vb{B}$ is given by an outer circle of radius
\begin{equation}\label{rplus}
    \sqrt{\tau(\vb{A}\vb{A}^*)\tau(\vb{B}\vb{B}^*)}.
\end{equation}
    If $\vb{B}$ is not bi-invariant, the spectral boundaries of $\vb{A}\vb{B}$ can be written as the image of the points $x \in \C$ satisfying
    \begin{equation}
        \partial_{\alpha} \mathcal{R}_{1,\vb{B}}
        \left(0,\mathfrak{t}_{\vb{A}}(x)S_{\vb{B}}(\mathfrak{t}_{\vb{A}}(x))\right)
        \left|S_{\vb{B}}(\mathfrak{t}_{\vb{A}}(x))\right|^2
        =
        \frac{1}{f_{\vb{A}}(x)},
    \end{equation}
    under the map
    \begin{equation}
        x \mapsto \frac{x}{S_{\vb{B}}(\mathfrak{t}_{\vb{A}}(x))},
    \end{equation}
    where $\mathcal{R}_{1,\vb{B}}$ denotes an appropriate branch and $S_{\vb{B}}$ is one $S$-transform of $\vb{B}$.

    \item \textbf{Boundary of type 2:} If $\tau(\vb{A}^{-1}) = 0$, then the spectral boundary of type 2 is an inner circle of radius
\begin{equation}\label{rmoins}
    \frac{1}{\sqrt{\tau\left((\vb{A}\vb{A}^*)^{-1}\right)\tau\left((\vb{B}\vb{B}^*)^{-1}\right)}}.
\end{equation}
    Otherwise, it is given by the points $z \in \C$ satisfying
    \begin{equation}\label{edge2}
        |z|^2
        h_{\vb{B}}\left(z\tau(\vb{A}^{-1})\right)
        =
        \frac{1}{
        \tau\left((\vb{A}\vb{A}^*)^{-1}\right)
        -
        |\tau(\vb{A}^{-1})|^2
        }.
    \end{equation}
\end{itemize}
\end{result}

\begin{remarks}
    \item In the case where $\vb{B}$ is a bi-invariant matrix, using the single-ring theorem~\cite{FeinbergZeeNPB501, guionnet2011single}, the formulas \eqref{rmoins} and \eqref{rplus} are already known.    
    \item If $\vb{B} = \alpha I_N + \vb{C}$, where $\vb{C}$ is a bi-invariant matrix and $\alpha\in \C^*$, we shall see in the examples below that the transforms can be computed rather explicitly. In this case, the boundaries of type $1$ are given by the points $z\in \C$ such that
    \begin{equation}\label{biinvariantcase}
        f_{\vb{A}}(z/\alpha)
        =
        \frac{|\alpha|^2}{r_{+,\vb{C}}^2}.
    \end{equation}

    Concerning the boundary of type $2$, we use equation \eqref{edge2} by taking
\begin{equation}
    h_{\vb{B}}(z)
    =
    \frac{1}{r_{-,\vb{C}}^2 - |z-\alpha|^2}.
\end{equation}
After several algebraic manipulations, we again obtain a circle with center
\begin{equation}
    c
    =
    \frac{
        \left|\tau(\vb{A}^{-1})\right|^2
    }{
        \tau\!\left((\vb{A}\vb{A}^*)^{-1}\right)
    }
    \alpha
\end{equation}
and radius squared
\begin{equation}
    R^2
    =
    \frac{
        \left|\tau(\vb{A}^{-1})\right|^2
        \left[
            \tau\!\left((\vb{A}\vb{A}^*)^{-1}\right)
            -
            \left|\tau(\vb{A}^{-1})\right|^2
        \right]
        |\alpha|^2
    }{
        \tau\!\left((\vb{A}\vb{A}^*)^{-1}\right)^2
    }
    +
    \frac{
        r_{-,\vb{C}}^2
    }{
        \tau\!\left((\vb{A}\vb{A}^*)^{-1}\right)
    }.
\end{equation}
Notice that, when $\alpha = 0$, we recover equation \eqref{rmoins}.

    \item If one studies $\vb{A}(I_N+\vb{B})$, with $\vb{A}$ Hermitian and $\vb{B}$ bi-invariant, then the boundary of type $1$ becomes particularly simple. The previous remark implies that the spectral boundary is given by the set of complex numbers $z\in \C$ satisfying
    \begin{equation}\label{caseH}
        r_{+,\vb{B}}^2 \Im\left(z \mathfrak{t}_{\vb{A}}(z)\right)
        =
        -\Im(z).
    \end{equation}
    Let us now examine the limit as $r_{+,\vb{B}} \to 0$. Denote by $\rho_{\vb{A}}$ the limiting spectral density of the real eigenvalues of $\vb{A}$, supported on the compact interval $[\lambda_{-},\lambda_{+}]$. In this regime, one finds
    \begin{equation}\label{traceMP}
        |\Im(z)|
        \underset{r_{+,\vb{B}}\to 0}{\sim}
        \Re(z)^2 r_{+,\vb{B}}^2\,\pi\,\rho_{\vb{A}}(\Re(z)).
    \end{equation}

    \item One can also obtain the corresponding statement for $\vb{B}\vb{A}$ by passing to the conjugate.

    \item Let us emphasize that this is not a generalization of the Hermitian case. Indeed, in the Hermitian setting, the natural object is rather $\vb{A}^{1/2}\vb{B}\vb{A}^{1/2}$.

    \item There are two types of spectral edges, and we note that the second one, described by equation~\eqref{edge2}, provides a rather universal spectral edge, which depends on the matrix $\vb{A}$ only through the two quantities $\tau(\vb{A}^{-1})$ and $\tau\left((\vb{A}\vb{A}^*)^{-1}\right)$.

    \item One should be careful with the interpretation of the result: we are not claiming that the curves given by the result will necessarily be spectral edges, but rather that the spectral edges of $\vb{A}\vb{B}$ must be described in this way.
\end{remarks}

In~\cite{bousseyroux3}, one can see that, in order to determine the spectral boundary of a matrix $\vb{M}$, it is enough to know the functions
\begin{equation}
    g \mapsto \partial_{\alpha}\mathcal{R}_{1,\vb{M}}(0,g)
    \qquad \text{and} \qquad
    g \mapsto \mathcal{R}_{2,\vb{M}}(0,g).
\end{equation}
Thus, by additivity of these transforms, if $\vb{M}=\vb{A}+\vb{B}$, it is enough to know
\begin{equation}
    g \mapsto \partial_{\alpha}\mathcal{R}_{1,\vb{A}}(0,g),
    \qquad
    g \mapsto \partial_{\alpha}\mathcal{R}_{1,\vb{B}}(0,g),
\end{equation}
and
\begin{equation}
    g \mapsto \mathcal{R}_{2,\vb{A}}(0,g),
    \qquad
    g \mapsto \mathcal{R}_{2,\vb{B}}(0,g).
\end{equation}
The question we now ask is what happens when $\vb{M}=\vb{A}\vb{B}$. More precisely, can one express the spectral boundary of $\vb{M}$ only in terms of the respective $\mathcal{R}$-transforms of $\vb{A}$ and $\vb{B}$? The following result shows that, in fact, it is enough to know
\begin{equation}
    g \mapsto \partial_{\alpha}\mathcal{R}_{1,\vb{A}^{-1}}(0,g),
    \qquad
    g \mapsto \partial_{\alpha}\mathcal{R}_{1,\vb{B}}(0,g),
\end{equation}
and
\begin{equation}
    g \mapsto \mathcal{R}_{2,\vb{A}}(0,g),
    \qquad
    g \mapsto \mathcal{R}_{2,\vb{B}}(0,g).
\end{equation}
The proof is in fact very simple. It is enough to observe that $z$ is a spectral boundary point of $\vb{A}\vb{B}$ if and only if $0$ is a spectral boundary point of $\vb{B}-z\vb{A}^{-1}$. One can then apply Result~2 of~\cite{bousseyroux3}, which gives the following result.

\begin{result}\label{th3}
   Let $\vb{A}$ be a large deterministic matrix and let $\vb{B}$ be a rotationally invariant random matrix. The spectral boundaries of $\vb{A}\vb{B}$ can be written as follows.

    \begin{itemize}
        \item \textbf{Boundary of type 1:} as $v/u$, where $(u,v)\in \C^2$ satisfies the following system of equations:
        \begin{equation}
            |u|^2
            \partial_{\alpha}\mathcal{R}_{1,\vb{B}}(0,u)
            +
            |v|^2
            \partial_{\alpha}\mathcal{R}_{1,\vb{A}^{-1}}(0,v)
            =
            1,
        \end{equation}
        and
        \begin{equation}
            u\mathcal{R}_{2,\vb{B}}(0,u)
            +
            v\mathcal{R}_{2,\vb{A}^{-1}}(0,v)
            =
            -1.
        \end{equation}
        Here $\mathcal{R}_{2,\vb{B}}$, $\mathcal{R}_{2,\vb{A}^{-1}}$, $\mathcal{R}_{1,\vb{B}}$, and $\mathcal{R}_{1,\vb{A}^{-1}}$ denote suitable branches.

        \item \textbf{Boundary of type 2:} the points $z\in \C$ such that
        \begin{equation}
    \mathcal{R}_{1,\vb{B}}
    \left(
        \alpha,
        -\bigl(\tau(\vb{B}) - z\tau(\vb{A}^{-1})\bigr)\alpha^2
    \right)
    +
    |z|^2
    \mathcal{R}_{1,\vb{A}^{-1}}
    \left(
        |z|\alpha,
        -\bigl(\tau(\vb{B}) - z\tau(\vb{A}^{-1})\bigr)\alpha^2
    \right)
    =
    -\frac{1}{\alpha}
    -
    \left|\tau(\vb{B}) - z\tau(\vb{A}^{-1})\right|^2\alpha
    +
    o(\alpha),
\end{equation}
        where $\mathcal{R}_{2,\vb{B}}$, $\mathcal{R}_{2,\vb{A}^{-1}}$, $\mathcal{R}_{1,\vb{B}}$, and $\mathcal{R}_{1,\vb{A}^{-1}}$ denote suitable branches.
    \end{itemize}
\end{result}

\begin{remarks}
    \item Once again, we are not claiming that the curves given by the result will necessarily be the spectral edges of $\vb{A}\vb{B}$, but rather that the spectral edges will be described in this way.

    \item By taking $\vb{A}=\vb{I}$, one simply recovers the spectral edges described in~\cite{bousseyroux3}.

    \item The advantage of presenting the result in this form is that the symmetry
    \begin{equation}
        \vb{A}\vb{B}
        \mapsto
        \vb{B}^{-1}\vb{A}^{-1}
    \end{equation}
    becomes rather transparent; it simply corresponds to inverting the eigenvalues.
    \item If $\vb{A}$ and $\vb{B}$ are bi-invariant matrices, one can recover equations~\eqref{rmoins} and~\eqref{rplus}.
\end{remarks}

The previous result highlights two new functions, namely
\begin{equation}
    g\mapsto
    \partial_\alpha\mathcal R_{1,\vb A^{-1}}(0,g)
    \qquad \text{and} \qquad
    g\mapsto
    \mathcal R_{2,\vb A^{-1}}(0,g).
\end{equation}
To determine the second one, one can use the usual techniques based on the $S$-transform; see~\cite{potters2020first}. The idea is to notice that
\begin{equation}
    \mathfrak{t}_{\vb{A}}(z) + \mathfrak{t}_{\vb{A}^{-1}}(1/z) + 1 = 0,
\end{equation}
and therefore, using~\eqref{definitionS_nonhermitian_alt}, one can deduce that if $S_{\vb A}$ comes from the use of equation~\eqref{definitionS_nonhermitian_alt}, then
\begin{equation}\label{inverseS}
    S_{\vb{A}^{-1}}(t)
    =
    \frac{1}{S_{\vb{A}}(-t-1)}
\end{equation}
is indeed an $S$-transform of $\vb A^{-1}$. Thus, using the relations~\eqref{relationRS}, one obtains one branch of
\begin{equation}
    g\mapsto
    \mathcal R_{2,\vb A^{-1}}(0,g).
\end{equation}

The theory developed here aims to provide efficient conjectures for computing the boundary of the spectrum. The reasoning above should be understood as a way, starting from an $S$-transform of $\vb A$, to potentially compute one for $\vb A^{-1}$, without being certain that all possible branches are controlled. One may then try this branch and obtain conjectures for the spectral boundary. We now state a new lemma, explained in the appendix, in this spirit. We use the word ``can'', since we currently have no control over all the possible branches associated with $\vb A^{-1}$.

\begin{lemma}\label{propinverse}
    Let $\vb A$ be a large random matrix. The expression
    \begin{equation}
        \partial_\alpha
        \mathcal R_{1,\vb A^{-1}}(0,g)
        =
        \frac{
            |\mathcal R_{2,\vb A^{-1}}(0,g)|^4
            \,
            \partial_\alpha
            \mathcal R_{1,\vb A}(0,u)
        }
        {
            1
            -
            |\mathcal R_{2,\vb A^{-1}}(0,g)|^2
            \left(
                1
                +
                2\Re\bigl(g\mathcal R_{2,\vb A^{-1}}(0,g)\bigr)
            \right)
            \partial_\alpha
            \mathcal R_{1,\vb A}(0,u)
        },
    \end{equation}
    where
    \begin{equation}
        u
        =
        -\mathcal R_{2,\vb A^{-1}}(0,g)
        -
        g\mathcal R_{2,\vb A^{-1}}(0,g)^2,
    \end{equation}
    can provide one branch of
    \begin{equation}
        g\mapsto
        \partial_\alpha\mathcal R_{1,\vb A^{-1}}(0,g),
    \end{equation}
    where $\mathcal R_{2,\vb A^{-1}}$ and $\mathcal R_{1,\vb A}$ denote suitable branches associated respectively with $\vb A^{-1}$ and $\vb A$.
\end{lemma}

\section{Examples}

\subsection{Basic random matrix ensembles}

We now introduce the basic random matrix ensembles that will be used throughout the paper to test and illustrate our results.

\begin{definition}[Complex Ginibre ensemble]\label{def:ginibre}
The matrix $\vb{G}$ denotes an $N\times N$ complex Ginibre matrix, i.e.\ a non-Hermitian random matrix with i.i.d.\ complex Gaussian entries of zero mean and variance $1/N$ \cite{Girko1985_CircularLaw_Eng}.
\end{definition}

\begin{definition}[Elliptic Ginibre ensemble]\label{def:elliptic_ginibre}
The elliptic Ginibre ensemble, denoted by $\vb{E}_\tau$, consists of $N\times N$ complex random matrices with Gaussian entries of zero mean and correlations
\begin{equation}
    \E[|X_{ij}|^2] = \frac{1}{N},
    \qquad
    \E[X_{ij}X_{ji}] = \frac{\tau}{N},
\end{equation}
for $1\le i\neq j\le N$, with $\tau\in[-1,1]$ \cite{sommers1988spectrum,girko1986elliptic}.
\end{definition}

In~\cite{bousseyroux1}, it was shown that
\begin{equation}
    \partial_{\alpha}\mathcal{R}_{1,\vb{E}_{\tau}}(0,g) = 1
    \qquad\text{and}\qquad
    \mathcal{R}_{2,\vb{E}_{\tau}}(0,g) = \tau g.
\end{equation}
Using equation~\eqref{defSS}, we see that there are then two $S$-transforms, given by the two determinations of the square root:
\begin{equation}\label{formula_elliptic}
    S_{\vb{E}_{\tau}}(g)
    =
    \pm \frac{1}{\sqrt{\tau g}}.
\end{equation}
Note that, for $\tau=1$, this allows one to define an $S$-transform for a Wigner matrix which was already known in \cite{rao2007multiplication}. When $\tau\to 0$, one approaches a Ginibre matrix, which is bi-invariant, and this is consistent with the convention~\eqref{Sbiinvariant}.

\begin{definition}[Haar unitary ensemble]\label{def:unitary}
The matrix $\vb{U}$ denotes an $N\times N$ random unitary matrix distributed according to the Haar measure on the unitary group $\mathrm{U}(N)$ \cite{mehta2004random}.
\end{definition}

The unitary matrix is bi-invariant, and hence its $S$-transform is $+\infty$.

\begin{definition}[Complex Wishart ensemble]\label{def:complex_wishart}
Let $\vb{A}$ be an $N\times T$ matrix with i.i.d.\ complex Gaussian entries of zero mean and unit variance. The matrix
\begin{equation}
    \vb{W}_q = \frac{1}{T}\vb{A}\vb{A}^*
\end{equation}
is called a complex Wishart matrix with aspect ratio $q=N/T$ \cite{wishart1928generalised}.
\end{definition}

Let us now consider the Wishart case. We saw in~\cite{bousseyroux3} that
\begin{equation}
    \partial_{\alpha}\mathcal{R}_{1,\vb{W}_{q}}(0,g)
    =
    \frac{q}{|1-qg|^2}
    \qquad\text{and}\qquad
    \mathcal{R}_{2,\vb{W}_{q}}(0,g)
    =
    \frac{1}{1-qg}.
\end{equation}
We then obtain a single $S$-transform,
\begin{equation}
    S_{\vb{W}_q}(g)
    =
    \frac{1}{1+qg},
\end{equation}
which is the $S$-transform of a Hermitian Wishart matrix; see~\cite{FirstCourseRMT}.
\begin{definition}[Ensemble 1]
   We consider the ensemble $\alpha I_N + \vb{B}$, where $\vb{B}$ is a bi-invariant matrix and $\alpha\in \C^*$. We set $\vb{M}=\alpha I_N+\vb{B}$.
\end{definition}

By additivity of $\mathcal{R}_1$ and $\mathcal{R}_2$, we have
\begin{equation}
    \partial_{\alpha}\mathcal{R}_{1,\vb{M}}(0,g)
    =
    \partial_{\alpha}\mathcal{R}_{1,\alpha\vb{1}}(0,g)
    +
    \partial_{\alpha}\mathcal{R}_{1,\vb{B}}(0,g)
    =
    r_{+,\vb{B}}^2,
\end{equation}
and
\begin{equation}
    \mathcal{R}_{2,\vb{M}}(0,g)
    =
    \mathcal{R}_{2,\alpha\vb{1}}(0,g)
    +
    \mathcal{R}_{2,\vb{B}}(0,g)
    =
    \alpha.
\end{equation}
Thus, the $S$-transform is simply
\begin{equation}
    S_{\vb{M}}(g)
    =
    \frac{1}{\alpha}.
\end{equation}

Using equation~\eqref{inverseS}, we obtain
\begin{equation}
    S_{\vb{M}^{-1}}(g) = \alpha,
\end{equation}
and therefore
\begin{equation}
    \mathcal{R}_{2, \vb{M}^{-1}}(0, g) = \frac{1}{\alpha}.
\end{equation}

\subsection{A cardioid curve}

\begin{result}
    Let $\vb{B}_1$ and $\vb{B}_2$ be independent and drawn from the same bi-invariant ensemble. Then, the outer boundary of the spectrum of $(I_N + \vb{B}_1)(I_N + \vb{B}_2)$ is given by the set of complex numbers $z \in \C$ such that
\begin{equation}\label{outer}
|1-z|^2
=
r_{+,\vb B}^2\left(1+|z|\right).
\end{equation}
When this makes sense, the circle of radius $r_{-, \vb{B}}^2 - 1$ is the inner spectral boundary.
\end{result}

\begin{remarks}
    \item The case $r_{+, \vb{B}} = 1$ yields the cardioid curve
\begin{equation}\label{eqlemniscate}
|z|^2-2\Re z
=
|z|.
\end{equation}

    \item Here, we do not assume that $\vb{B}$ is a Ginibre matrix. The result is universal in $\vb{B}$ as long as it is bi-invariant. We therefore generalize the cardioid result presented in~\cite{burda2011multiplication}.

    \item Figure~\ref{fig:lemniscate} illustrates this result. The middle and right panels display the same outer boundary, highlighting the universal character of the spectral edges in the bi-invariant case.
\end{remarks}
\begin{proof}
    For the boundaries of type $1$, it is enough to apply Result~\ref{th3} using formula~\eqref{formulaR}. 
    Let us denote $\vb{A}_1 = I_N + \vb{B}_1$ and $\vb{A}_2 = I_N + \vb{B}_2$. For the boundaries of type $2$, one has to be careful since $\tau(\vb{A}_1^{-1}) = 0$. We must therefore apply the bi-invariant case for the boundaries of type $2$ in Result~\ref{th3}. We then obtain a circle of radius
    \begin{equation}
        \frac{1}{\sqrt{\tau\left([\vb{A}_1\vb{A}_1^*]^{-1}\right)
        \tau\left([\vb{A}_2\vb{A}_2^*]^{-1}\right)}}
        =
        \frac{1}{\tau\left([\vb{A}_1\vb{A}_1^*]^{-1}\right)}.
    \end{equation}

    Moreover,
    \begin{equation}
        \tau\left([\vb{A}_1\vb{A}_1^*]^{-1}\right)
        =
        \tau\left([(I_N + \vb{B}_1)(I_N + \vb{B}_1)^*]^{-1}\right)
        =
        \frac{1}{r_{-, \vb{B}}^2 - 1},
    \end{equation}
    using the bi-invariance of $\vb{B}_1$. We then obtain the result.
\end{proof}

\begin{figure}[h!]
    \centering
\includegraphics[width=\textwidth]{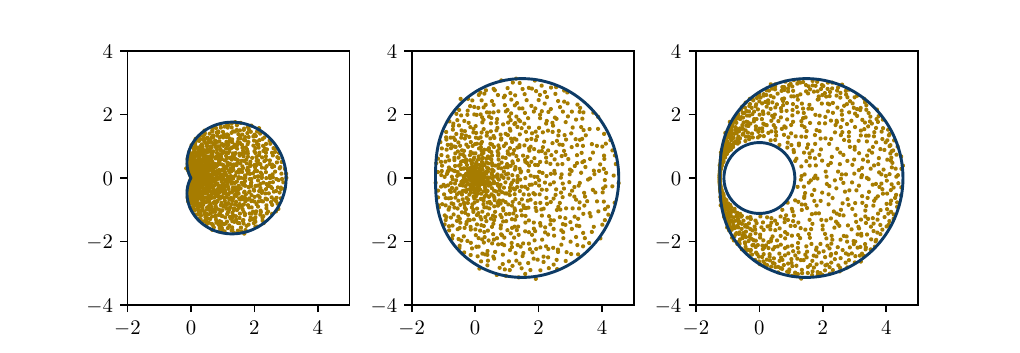}
    \caption{Empirical eigenvalues.
Left: $(I_N + \vb{G}_1)(I_N + \vb{G}_2)$, where $\vb{G}_1$ and $\vb{G}_2$ are Ginibre matrices ($N=1000$); the cardioid curve
\eqref{eqlemniscate} gives the theoretical spectral boundary.
Center: $(I_N + \sigma\vb{G}_1)(I_N + \sigma\vb{G}_2)$ with $\sigma=1.1$, where $\vb{G}_1$ and $\vb{G}_2$ are complex Ginibre matrices
($N=1000$); the outer boundary is given by \eqref{outer} with $r_{+,\vb{B}}=\sigma$.
Right: $(I_N + \sigma\vb{U}_1)(I_N + \sigma\vb{U}_2)$ with $\sigma=1.1$, where $\vb{U}_1$ and $\vb{U}_2$ are random unitary matrices
($N=1000$); in addition to the outer boundary \eqref{outer}, an inner spectral edge appears
at radius $\sigma^2-1$.}
    \label{fig:lemniscate}
\end{figure}

\subsection{A fish}

\begin{result}
Let $\vb{W}_q$ be a complex Wishart matrix defined in Definition \ref{def:complex_wishart} and let $\vb{B}$ be a large bi-invariant matrix such that $r_{-, \vb{B}} = 0$ and $r_{+, \vb{B}} = 1$. The spectral boundary of $\vb{M}$, defined by
\begin{equation}
    \vb{M} = \vb{W}_q(I_N + \sigma \vb{B}),
\end{equation}
is given by the set of points $z \in \C$ satisfying
\begin{equation}\label{theo}
    \sigma^2\,\Im\left(z\frac{z + q - 1 \pm \sqrt{(z + q - 1)^2 - 4qz}}{2q}\right) = - (1-\sigma^2)\Im(z).
\end{equation}
\end{result}

\begin{remarks}
    \item Figure~\ref{fig:aile_avion} illustrates the case where $\vb{B}$ is a complex Ginibre matrix. For small values of $\sigma$, using equation~\eqref{traceMP}, the spectral boundaries are described by the union of the curves $x \mapsto \sigma^2 x^2 \pi \rho(x)$ and $x \mapsto -\sigma^2 \pi x^2 \rho(x)$, where $\rho$ denotes the Marchenko--Pastur density.
\end{remarks}

\begin{proof}
It suffices to use equation \eqref{caseH} together with the expression of $\mathfrak{t}_{\vb{W}_q}$ computed in \cite{potters2020first}:
\begin{equation}
    \mathfrak{t}_{\vb{W}_q} = \frac{z + q - 1 \pm \sqrt{(z + q - 1)^2 - 4qz}}{2q} - 1.
\end{equation}
\end{proof}

\begin{figure}[h!]
    \centering
    \includegraphics{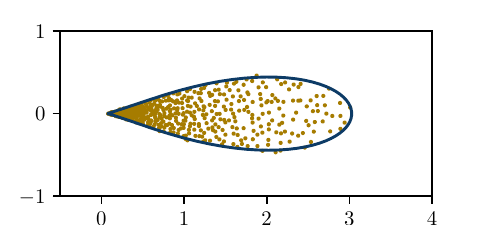}
    \caption{Complex eigenvalues of a matrix of size $N = 500$ given by $\vb{W}_q(I_N + \sigma \vb{G})$, where $\vb{W}_q$ is a complex Wishart matrix with parameter $q = 0.5$, $\sigma = 0.5$, and $\vb{G}$ is a complex Ginibre matrix. The theoretical spectral boundary is also shown using the equation \eqref{theo}.}
    \label{fig:aile_avion}
\end{figure}

\subsection{Product by an elliptic matrix}

We are interested here in the multiplication by an elliptic matrix. We apply result~\ref{propcentre}.

\begin{result}
    Let $\vb{A}$ be a large rotationally invariant random matrix, centered in the sense that $\tau(\vb{A})=0$, and let $\vb{E}_{\tau}$ be an elliptic matrix with parameter $\tau$. Then the spectrum of $\vb{A}\vb{E}_{\tau}$ is a ring and, for all $z\in \C$, one has
    \begin{equation}
        S_{\vb{A}\vb{A}^*}\left(F_{\vb{A}\vb{E}_{\tau}}(|z|)-1\right)
        =
        \frac{F_{\vb{A}\vb{E}_{\tau}}(|z|)}{|z|^2}.
    \end{equation}

    Moreover,
    \begin{equation}\label{formulaoverlap_elliptic}
        \mathcal{O}_{\vb{A}\vb{E}_{\tau}}(z)
        =
        \frac{
            F_{\vb{A}\vb{E}_{\tau}}(|z|)
            \left(1-F_{\vb{A}\vb{E}_{\tau}}(|z|)\right)
        }{|z|^2}.
    \end{equation}
\end{result}

\begin{remarks}
    \item What is remarkable here is that the parameter $\tau$ does not appear. In~\cite{burda2010spectrum}, it was shown that a product of elliptic matrices has the same limiting spectral distribution as a product of Ginibre matrices. We generalize this result.
\end{remarks}

\begin{proof}
    We apply result~\ref{propcentre} to the product $\vb{E}_{\tau}\vb{A}$: the first factor $\vb{E}_{\tau}$ has the same distribution as $-\vb{E}_{\tau}$, while the second factor $\vb{A}$ is rotationally invariant and centered. The matrices $\vb{A}\vb{E}_{\tau}$ and $\vb{E}_{\tau}\vb{A}$ have the same eigenvalues. Moreover, it is known that the singular values of an elliptic matrix with parameter $\tau$ do not depend on $\tau$; see for instance~\cite{bousseyroux1}. Thus, $S_{\vb{E}_{\tau}\vb{E}_{\tau}^*}$ is simply the $S$-transform of a Wishart matrix with parameter $q=1$.
\end{proof}

We now want to apply Result~\ref{result_boundaries}.

\begin{result}\label{propbonhomme}
    Let
    \begin{equation}
        \vb{D}
        =
        \diag(1,\ldots,1,-1,\ldots,-1)
    \end{equation}
    be a diagonal matrix with half of its entries equal to $1$ and the other half equal to $-1$. Let $\vb{E}_{\tau}$ be an elliptic matrix with parameter $\tau$, and set
    \begin{equation}
        \vb{M}
        =
        (\sigma I_N+\vb{D})\vb{E}_{\tau}.
    \end{equation}
    Then the spectral boundaries of $\vb{M}$ are given by the following procedure. First find the points $x\in \C$ satisfying
 \begin{equation}
    |\tau|
    \left|    
        \frac{\sigma+1}{x-(1+\sigma)}
        +
        \frac{\sigma-1}{x-(\sigma-1)}
    \right|
    =
    \frac{|\sigma+1|^2}{|x-(1+\sigma)|^2}
    +
    \frac{|\sigma-1|^2}{|x-(\sigma-1)|^2}.
\end{equation}
    The boundary is then obtained as the image of this curve under the map
    \begin{equation}
        x
        \mapsto
        x
        \sqrt{
        \tau
        \frac{1}{2}
        \left(
        \frac{\sigma+1}{x-(1+\sigma)}
        +
        \frac{\sigma-1}{x-(\sigma-1)}
        \right)
        }.
    \end{equation}
\end{result}

\begin{remarks}
    \item Figure~\ref{fig:bonhomme} illustrates the result for $\tau=0.9$ and $\sigma=0.5$.
\end{remarks}

\begin{proof}
    It is enough to apply Result~\ref{result_boundaries}, using the facts that
    \begin{equation}
        S_{\vb{E}_{\tau}}(t)
        =
        \pm \frac{1}{\sqrt{\tau t}},
    \end{equation}
    and
    \begin{equation}
        \partial_{\alpha}
        \mathcal{R}_{1,\vb{E}_{\tau}}(0,g)
        =
        1.
    \end{equation}
\end{proof}

\begin{figure}
    \centering
    \includegraphics[width=0.5\linewidth]{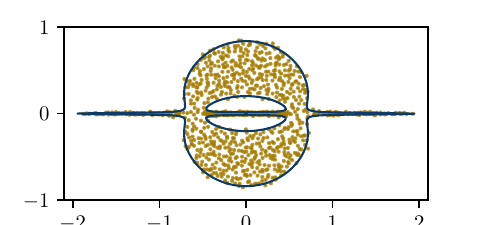}
    \caption{Empirical complex eigenvalues of $(\sigma I_N+\vb{D})\vb{E}_{\tau}$, where $\vb{D}
=
\diag(1,\ldots,1,-1,\ldots,-1)$ and $\vb{E}_{\tau}$ is an elliptic matrix with $\tau=0.9$ and $\sigma=0.5$. The theoretical boundary is drawn using result~\ref{propbonhomme}.}
    \label{fig:bonhomme}
\end{figure}

\subsection{The inverse Wigner matrix}

This section is exploratory. Let us return to the inverse of a Wigner matrix $\vb{W}$. In the case $\tau=1$, Eq.~\eqref{formula_elliptic} gives
\begin{equation}
    S_{\vb{W}}(g)=\pm \frac{1}{\sqrt{g}}.
\end{equation}
Using the inverse relation~\eqref{inverseS}, we are led to the following branches for the $S$-transform of $\vb{W}^{-1}$:
\begin{equation}
    S_{\vb{W}^{-1}}(t)=\pm \sqrt{-t-1}.
\end{equation}

In particular, as $t\to 0$, these branches converge to $\pm i$. At first sight, this may seem paradoxical. Indeed, in the Hermitian setting, one usually has, for a matrix $\vb{A}$ with a well-defined first moment,
\begin{equation}
    S_{\vb{A}}(x)
    \underset{x\to 0}{\longrightarrow}
    \frac{1}{\tau(\vb{A})},
\end{equation}
see~\cite{FirstCourseRMT}. A naive application of this formula would suggest
\begin{equation}
    \tau(\vb{W}^{-1})=\pm i,
\end{equation}
which cannot be interpreted as an ordinary trace moment: the value is not unique, it is purely imaginary although $\vb{W}^{-1}$ is Hermitian, and the integral defining the inverse moment is singular at the origin.

One possible interpretation is that these quantities are regularized. Indeed, in the large-$N$ limit, the spectrum of $\vb{W}$ is described by the semicircle law
\begin{equation}
    \rho(\lambda)
    =
    \frac{1}{2\pi}\sqrt{4-\lambda^2}\,
    \mathbf{1}_{[-2,2]}(\lambda).
\end{equation}
The ordinary inverse moment
\begin{equation}
    \int_{-2}^{2}\frac{\rho(\lambda)}{\lambda}\,d\lambda
\end{equation}
is not canonically defined. If instead one regularizes the pole by moving it away from the real axis, one obtains
\begin{equation}
    M_\varepsilon^{\pm}
    =
    \int_{-2}^{2}
    \frac{\rho(\lambda)}{\lambda\mp i\varepsilon}\,d\lambda,
\end{equation}
and therefore
\begin{equation}
    M_\varepsilon^{\pm}
    \underset{\varepsilon\to 0^+}{\longrightarrow}
    \pm i.
\end{equation}

The relation~\eqref{relationRS} gives
\begin{equation}
    R(g)
    =
    \frac{1}{S_{\vb{W}^{-1}}(gR(g))},
\end{equation}
and hence
\begin{equation}
    gR(g)^3+R(g)^2+1=0.
\end{equation}
The three solutions are given by
\begin{equation}
R_k(g)
=
-\frac{1}{3g}
+
\omega^k
\sqrt[3]{
-\frac{27g^2+2}{54g^3}
+
\sqrt{\frac{27g^2+4}{108g^4}}
}
+
\omega^{-k}
\sqrt[3]{
-\frac{27g^2+2}{54g^3}
-
\sqrt{\frac{27g^2+4}{108g^4}}
},
\qquad k=0,1,2,
\end{equation}
where
\begin{equation}
    \omega=e^{2\pi i/3}.
\end{equation}
The two branches which remain finite at $g=0$ satisfy
\begin{equation}
    R_{\pm}(g)
    =
    \pm i
    +
    \frac{1}{2}g
    +
    O(g^2).
\end{equation}

Therefore, the quantities obtained from this expansion should not be read as ordinary moments of $\vb{W}^{-1}$, but rather as regularized free cumulants. As for the usual $R$-transform, one may expect these regularized free cumulants to remain additive.

This is surprising because the usual transform of $\vb{W}^{-1}$ cannot be defined through the standard moment expansion. Indeed, the eigenvalues of $\vb{W}^{-1}$ are unbounded in the large-$N$ limit, since the spectrum of $\vb{W}$ reaches the origin. Thus the usual route, based on ordinary inverse moments and on an expansion at infinity, breaks down. What remains to be understood is the precise meaning of this $R$-transform, which cannot be obtained from the standard Hermitian construction, but which nevertheless appears naturally through the analytic continuation of the non-Hermitian $S$-transform formalism.

\newpage
\paragraph*{Acknowledgements.} 

This research was conducted within the Econophysics \& Complex Systems Research Chair, under the aegis of the Fondation du Risque, the Fondation de l’Ecole polytechnique, the Ecole polytechnique, and Capital Fund Management.

\bibliographystyle{plain}
\bibliography{References.bib}

\appendix

\newpage

\section{APPENDICES}
\subsection{Reminder on the replica method for random matrices}
\label{sec:replica}
Let $\vb{C}$ be an $N\times N$ Hermitian matrix with strictly positive real eigenvalues. The standard complex Gaussian identities yield
\begin{equation}
    \frac{1}{\pi^N}
    \int_{\C^N}
    \phi_a\,\overline{\phi_b}\,
    e^{-\langle \phi,\vb{C}\phi\rangle}\,
    \mathrm{d}\phi
    \;=\;
    \frac{[\vb{C}^{-1}]_{ab}}{\det(\vb{C})},
    \qquad 1\leq a,b\leq N,
\end{equation}
where $d\phi$ denotes the Lebesgue measure on $\C^N$.
As a consequence,
\begin{equation}
    [\vb{C}^{-1}]_{ab}
    \;=\;
    \frac{\displaystyle
    \int_{\C^N}
    \phi_a\,\overline{\phi_b}\,
    e^{-\langle \phi,\vb{C}\phi\rangle}\,
    d\phi}
    {\displaystyle
    \int_{\C^N}
    e^{-\langle \phi,\vb{C}\phi\rangle}\,
    \mathrm{d}\phi}.
\end{equation}

To introduce replicas, we multiply both the numerator and the denominator by
\(
\bigl(\pi^{-N}\!\int_{\C^N} e^{-\langle \phi,\vb{C}\phi\rangle}\,\mathrm{d}\phi\bigr)^{n-1},
\)
and introduce $n$ independent copies
$\phi^{(1)},\dots,\phi^{(n)}\in\C^N$.
This gives
\begin{equation}
    [\vb{C}^{-1}]_{ab}
    \;=\;
    \det(\vb{C})^{\,n}
    \int_{(\C^N)^n}
    [\phi^{(1)}]_a\,\overline{[\phi^{(1)}]_b}\,
    \exp\!\Bigl(
        -\sum_{k=1}^n
        \langle \phi^{(k)},\vb{C}\phi^{(k)}\rangle
    \Bigr)\,
    \frac{\mathrm{d}\phi^{(1)}\cdots \mathrm{d}\phi^{(n)}}{\pi^{Nn}}.
\end{equation}
This identity holds for all integers $n\ge 1$.
The replica method then consists in formally taking the limit $n\to 0$,
using $\det(\vb{C})^{\,n}\to 1$, which leads to
\begin{equation}\label{eq:replica}
    [\vb{C}^{-1}]_{ab}
    \;=\;
    \lim_{n\to 0}
    \int_{(\C^N)^n}
    [\phi^{(1)}]_a\,\overline{[\phi^{(1)}]_b}\,
    \exp\!\Bigl(
        -\sum_{k=1}^n
        \langle \phi^{(k)},\vb{C}\phi^{(k)}\rangle
    \Bigr)\,
    \frac{\mathrm{d}\phi^{(1)}\cdots \mathrm{d}\phi^{(n)}}{\pi^{Nn}}.
\end{equation}

Formula~\eqref{eq:replica} provides a convenient representation for resolvent entries in Hermitian settings.
It is particularly well suited for saddle-point analysis and behaves naturally under sums of independent matrices.

We now consider a $2N\times 2N$ Hermitian matrix $\vb{C}$ with block structure
\begin{equation}
    \vb{C} \;=\;
    \begin{pmatrix}
        \vb{C}_{11} & \vb{C}_{12} \\
        \vb{C}_{12}^* & \vb{C}_{22}
    \end{pmatrix},
    \qquad
    \vb{C}_{11},\vb{C}_{22}, \vb{C}_{12}\in\C^{N\times N}.
\end{equation}
Writing the inverse in block form,
\begin{equation}
    \vb{C}^{-1} \;=\;
    \begin{pmatrix}
        \vb{G}_{11} & \vb{G}_{12} \\
        \vb{G}_{21}^* & \vb{G}_{22}
    \end{pmatrix},
\end{equation}
the replica representation yields, for $i,j\in\{1,2\}$,
\begin{multline}\label{eq:block-replica}
    \vb{G}_{ij}
    \;=\;
    \lim_{n\to 0}
    \int_{\phi_1^{(k)},\,\phi_2^{(k)}\in\C^N}
    \phi_i^{(1)}\,\left(\phi_j^{(1)}\right)^*
    \\
    \times
    \exp\!\Biggl(
        -\sum_{k=1}^n\Bigl[
            \langle \phi_1^{(k)},\vb{C}_{11}\phi_1^{(k)}\rangle
            + \langle \phi_2^{(k)},\vb{C}_{22}\phi_2^{(k)}\rangle
            + 2\,\Re\!\bigl\langle \phi_1^{(k)},\vb{C}_{12}\phi_2^{(k)}\bigr\rangle
        \Bigr]
    \Biggr)
    \,
    \frac{\mathrm{d}\phi_1^{(1)}\!\cdots \mathrm{d}\phi_1^{(n)}\,
          \mathrm{d}\phi_2^{(1)}\!\cdots \mathrm{d}\phi_2^{(n)}}{\pi^{2Nn}}.
\end{multline}
We now specialize to the matrix
\begin{equation}
    \vb{C} =
    \begin{pmatrix}
        \omega\,\vb{I}_N & z\,\vb{I}_N - \vb{M} \\
        (z\,\vb{I}_N - \vb{M})^* & \omega\,\vb{I}_N
    \end{pmatrix},
\end{equation}
where $z\in\C$, $\vb{M}$ is a non-Hermitian $N\times N$ matrix, and $\omega>0$. We assume that $\omega$ is large, so that $\vb{C}$ is a Hermitian matrix with strictly positive eigenvalues. In this case, the representation~\eqref{eq:block-replica} becomes
\begin{multline}\label{bigformula}
    \vb{G}_{ij}(\omega,z)
    \;=\;
    \lim_{n\to 0}
    \int_{\phi_1^{(k)},\,\phi_2^{(k)}\in\C^N}
    \phi_i^{(1)}\,\left(\phi_j^{(1)}\right)^*
    \\
    \times
    \exp\!\Biggl(
        -\sum_{k=1}^n\Bigl[
            \omega \|\phi_1^{(k)}\|^2
            + \omega \|\phi_2^{(k)}\|^2
            + 2\,\Re\!\bigl\langle \phi_1^{(k)},(z-\vb{M})\phi_2^{(k)}\bigr\rangle
        \Bigr]
    \Biggr)
    \,
    \frac{\mathrm{d}\phi_1^{(1)}\!\cdots \mathrm{d}\phi_1^{(n)}\,
          \mathrm{d}\phi_2^{(1)}\!\cdots \mathrm{d}\phi_2^{(n)}}{\pi^{2Nn}}.
\end{multline}

\subsection{Derivation of the main result \ref{mainresult}}\label{sec:derivation}

The proof will be very close to the paper \cite{bousseyroux3}. To simplify notation, we work with a single replica $n=1$ and absorb the overall exponential factor at the end.

Given an $N\times N$ matrix $\vb{M}$ and spectral parameter $z\in\C$, embed $\vb{M}$ into the $2N\times2N$ block-Hermitian matrix
\begin{equation}
\vb{H}(z)\;:=\;
\begin{pmatrix}
0 & z\vb{1}-\vb{M}\\
\overline{z}\vb{1}-\vb{M}^* & 0
\end{pmatrix}.
\end{equation}
For  $\omega\in\C$, define the block resolvent
\begin{equation}
\vb{G}_{\vb{M}}(\omega,z)\;:=\;\bigl[\omega \vb{1}-\vb{H}(z)\bigr]^{-1}
\;=\;
\begin{pmatrix}
\vb{G}_{11} & \vb{G}_{12}\\
\vb{G}_{21} & \vb{G}_{22}
\end{pmatrix},
\end{equation}
and the block-trace scalars
\begin{equation}
\g_{ij}^{N}(\omega,z)\;:=\;\frac{1}{N}\Tr\,\vb{G}_{ij}(\omega,z),\qquad i, j\in\{1,2\}.
\end{equation}
The \emph{quaternionic} (matrix-valued) resolvent is then defined as
\begin{equation}\label{eq:quaternionic-resolvent}
\mathcal{G}_{\vb{M}}^{N}(\omega,z)
\;:=\;
\begin{pmatrix}
\g_{11}^{N}(\omega,z) & \g_{12}^{N}(\omega,z)\\
\g_{21}^{N}(\omega,z) & \g_{22}^{N}(\omega,z)
\end{pmatrix}.
\end{equation}
Its large-$N$ limit is given by
\begin{equation}\label{defG}
    \mathcal{G}_{\vb{M}}(\omega,z)
    \;=\;
    \lim_{N \to \infty} \mathcal{G}_{\vb{M}}^{N}(\omega,z)
    \;=\;
    \begin{pmatrix}
        \g_1(\omega, z) & \g_2(\omega, z) \\
        \overline{\g_2}(\omega, z) & \g_1(\omega, z)
    \end{pmatrix},
\end{equation}
where
\begin{align}
\g_1(\omega, z)
&=\;
\tau\!\left[\omega\bigl(\omega^2 \vb{1}-(z \vb{1}-\vb{M})(z \vb{1}-\vb{M})^*\bigr)^{-1}\right],
\\[0.5em]
\g_2(\omega, z)
&=\;
-\,\tau\!\left[\bigl(\omega^2 \vb{1}-(z \vb{1}-\vb{M})(z \vb{1}-\vb{M})^*\bigr)^{-1}(z \vb{1}-\vb{M})^*\right].
\end{align}
and $\tau := \lim_{N \to \infty} \frac{1}{N}\,\mathrm{Tr}$ denotes the normalized trace in the large-$N$ limit. 

We apply formula~\eqref{bigformula} in the case $\vb{M} = \vb{A}\vb{U}\vb{B}\vb{U}^*$:
\begin{multline}
    [\mathcal{G}_{\vb{M}}(\omega, z)]_{i j} 
    \;\propto\;
    \int_{\phi_1,\,\phi_2 \in \C^N} 
        \phi_i \,\phi_j^* \,
        \exp\!\left\{-\omega \|\phi_1\|^2 \;-\; \omega \|\phi_2\|^2 \;-\; 2\,\Re\!\bigl(z\,\langle \phi_1, \phi_2\rangle\bigr)\right\} \\
    \times 
        \exp\!\left\{2\,\Re\!\bigl(\langle \phi_1, \vb{A}\vb{U}\vb{B}\vb{U}^*\,\phi_2\rangle\bigr)\right\}
    \,\mathrm{d}\phi_1 \,\mathrm{d}\phi_2 .
\end{multline}
Now take expectation with respect to $\vb{U}$ using the definition of $\mathcal{H}_{\vb{B}}$, the main object of \cite{bousseyroux1}.
\begin{multline}
    \E_{\vb{U}}\!\bigl([\vb{G}_{\vb{M}}(\omega, z)]_{i j}\bigr) 
    \;\propto\;
    \int_{\phi_1,\,\phi_2 \in \C^N} 
        \phi_i \,\phi_j^* \,
        \exp\!\left\{-\omega \|\phi_1\|^2 \;-\; \omega \|\phi_2\|^2 \;-\; 2\,\Re\!\bigl(z\,\langle \phi_1, \phi_2 \rangle\bigr)\right\} \\
    \times 
        \exp\!\left\{2N \,\mathcal{H}_{\vb{B}}\!\Bigl(
            \tfrac{\|\vb{A}^*\phi_1\|}{\sqrt{N}}\;\tfrac{\|\phi_2\|}{\sqrt{N}},\;
            \tfrac{\langle \phi_1, \vb{A}\phi_2\rangle}{N}
        \Bigr)\right\}
    \,\mathrm{d}\phi_1 \,\mathrm{d}\phi_2 .
\end{multline}
Introduce the Lagrange multipliers:
\begin{equation}
    C \;=\; \frac{\|\vb{A}^*\phi_1\|^2}{N}, 
    \qquad 
    D \;=\; \frac{\|\phi_2\|^2}{N}, 
    \qquad 
    E \;=\; \frac{\langle \phi_1, \vb{A}\phi_2\rangle}{N},
\end{equation}
and recall the representation of the Dirac delta as an integral over the imaginary axis:
\begin{equation}
    \delta(x)
    \;=\;
    \int_{-i\infty}^{i\infty} \frac{e^{-x y}}{2\pi i}\,dy.
\end{equation}
Then
\begin{multline}
    \E_{\vb{U}}\!\bigl([\vb{G}_{\vb{M}}(\omega, z)]_{i j}\bigr) 
    \;\propto\;
    \int 
    \int_{\phi_1,\,\phi_2 \in \C^N} 
        \phi_i \,\phi_j^* \,
        \exp\!\left\{-\omega \|\phi_1\|^2 \;-\; \omega \|\phi_2\|^2 \;-\; 2\,\Re\!\bigl(z\,\langle \phi_1, \phi_2\rangle\bigr)\right\} \\
    \times 
        \exp\!\left\{2N \,\mathcal{H}_{\vb{B}}(\sqrt{C D},\,E)\right\} \\
    \times 
        \exp\!\left\{
            c\Bigl(\tfrac{\langle\phi_1, \vb{A}\vb{A}^* \phi_1\rangle}{N} - C\Bigr)
            \;+\;
            d\Bigl(\tfrac{\|\phi_2\|^2}{N} - D\Bigr)
            \;+\;
            \Re\!\Bigl[\overline{e}\,\bigl(\tfrac{\langle \phi_1, \vb{A}\phi_2\rangle}{N} - E\bigr)\Bigr]
        \right\} \\
    \mathrm{d}\phi_1 \,\mathrm{d}\phi_2 \,
    \mathrm{d}c \,\mathrm{d}d \,\mathrm{d}\Re(e) \,\mathrm{d}\Im(e)\,\mathrm{d}C \,\mathrm{d}D \,\mathrm{d}\Re(E) \,\mathrm{d}\Im(E).
\end{multline}

Make the change of variables $c \mapsto cN$, $d \mapsto dN$, $e \mapsto 2Ne$:
\begin{multline}\label{ap}
    \E\!\bigl([\vb{G}_{\vb{M}}(\omega, z)]_{i j}\bigr) 
    \;\propto\;
    \int 
    \int_{\phi_1,\,\phi_2 \in \C^N} 
        \phi_i \,\phi_j^* \,
        \exp\!\left\{-\langle\phi_1, (\omega \vb{1} - c \vb{A}\vb{A}^*)\phi_1\rangle \;-\; (\omega - d)\|\phi_2\|^2 \;-\; 2\,\Re\langle\phi_1,(z\vb{1}-\overline{e}\vb{A})\phi_2\rangle\right\} \\
    \times 
        \exp\!\left\{2N \,\mathcal{H}_{\vb{B}}(\sqrt{C D},\,E)\right\}
        \exp\!\left\{-N\,c\,C \;-\; N\,d\,D \;-\; 2N\,\Re\!\bigl[e\,E\bigr]\right\}
    \\\,\mathrm{d}\phi_1 \,\mathrm{d}\phi_2 \,
    \mathrm{d}c \,\mathrm{d}d \,\mathrm{d}\Re(e) \,\mathrm{d}\Im(e)\,\mathrm{d}C \,\mathrm{d}D \,\mathrm{d}\Re(E) \,\mathrm{d}\Im(E).
\end{multline}
We now perform the Gaussian integrals over $\phi_1$ and $\phi_2$. 
Equation~\eqref{ap} then yields
\begin{equation}\label{ap2}
    \E_{\vb{U}}\!\bigl([\vb{G}_{\vb{M}}(\omega, z)]_{i j}\bigr)
    \;\approx\footnote{The symbol $\approx$ is inherently imprecise due to the use of the replica method. The result we aim to explain assumes that $\approx$ preserves the block operator $\mathcal{G}$ defined in Eq.~\eqref{defG}.}
    [\vb{H}_{\vb{A}}(\omega, z)]_{i, j},
\end{equation}
where
\begin{equation}
    \vb{H}_{\vb{A}}(\omega, z)
    =
    \begin{pmatrix}
        \omega \vb{1} - c \vb{A}\vb{A}^* & z - \overline{e}\vb{A} \\
        (z - \overline{e}\vb{A})^* & \omega - d
    \end{pmatrix}^{-1}.
\end{equation}
The parameters $c,d,e,C,D,E$ are chosen so as to optimize
\begin{equation}\label{eqaderivee}
    2\,\mathcal{H}_{\vb{B}}(\sqrt{C D},\,E)
    -
    c\,C
    -
    d\,D
    -
    \overline{e}\,E
    -
    e\,\overline{E}
    +
    \log\!\det \vb{H}_{\vb{A}}(\omega,z).
\end{equation}
Taking derivatives of \eqref{eqaderivee} with respect to $C,D,E$ gives
\begin{equation}\label{system1}
    \left\{
    \begin{aligned}
        \sqrt{\frac{D}{C}}\,
        \partial_{\alpha}\mathcal{H}_{\vb{B}}(\sqrt{CD},\,E)
        &= c, \\[0.4em]
        \sqrt{\frac{C}{D}}\,
        \partial_{\alpha}\mathcal{H}_{\vb{B}}(\sqrt{CD},\,E)
        &= d, \\[0.4em]
        2\,\partial_{\beta}\mathcal{H}_{\vb{B}}(\sqrt{CD},\,E)
        &= \overline{e}.
    \end{aligned}
    \right.
\end{equation}

Using the identity
\begin{equation}
    \partial_t \log\!\det\!\bigl(\vb{C}(t)\bigr)
    =
    \Tr\!\bigl(\vb{C}(t)^{-1}\vb{C}'(t)\bigr),
\end{equation}
the derivatives of \eqref{eqaderivee} with respect to $c,d,\overline{e}$ yield
\begin{equation}\label{system2}
    \left\{
    \begin{aligned}
        D
        &=
        \tau\!\left[
            \bigl(\omega - c\vb{A}\vb{A}^*\bigr)
            \Bigl(
                \bigl(\omega - c\vb{A}\vb{A}^*\bigr)(\omega-d)
                -
                |z-\overline{e}\vb{A}|^2
            \Bigr)^{-1}
        \right],
        \\[0.6em]
        E
        &=
        -\tau\!\left[
            \vb{A}(z-\overline{e}\vb{A})^*
            \Bigl(
                \bigl(\omega - c\vb{A}\vb{A}^*\bigr)(\omega-d)
                -
                |z-\overline{e}\vb{A}|^2
            \Bigr)^{-1}
        \right],
        \\[0.6em]
        C
        &=
        \tau\!\left[
            \vb{A}\vb{A}^*(\omega-d)
            \Bigl(
                \bigl(\omega - c\vb{A}\vb{A}^*\bigr)(\omega-d)
                -
                |z-\overline{e}\vb{A}|^2
            \Bigr)^{-1}
        \right]
    \end{aligned}
    \right.
\end{equation} where $|X|^2$ denotes $XX^*$.

Finally, to summarize, using~\eqref{ap2}, we obtain
\begin{equation}
    \g_{1,\vb{M}}(\omega,z)
    =
    \tau\!\left[
        (\omega-d)
        \Bigl(
            \bigl(\omega - c\vb{A}\vb{A}^*\bigr)(\omega-d)
            -
            |z-\overline{e}\vb{A}|^2
        \Bigr)^{-1}
    \right],
\end{equation}
\begin{equation}
    \g_{1,\vb{M}}(\omega,z)
    =
    \tau\!\left[
        \bigl(\omega - c\vb{A}\vb{A}^*\bigr)
        \Bigl(
            \bigl(\omega - c\vb{A}\vb{A}^*\bigr)(\omega-d)
            -
            |z-\overline{e}\vb{A}|^2
        \Bigr)^{-1}
    \right].
\end{equation}
\begin{equation}
    \g_{2,\vb{M}}(\omega,z)
    =
   -\tau\!\left[
        (z-\overline{e}\vb{A})^*
        \Bigl(
            \bigl(\omega - c\vb{A}\vb{A}^*\bigr)(\omega-d)
            -
            |z-\overline{e}\vb{A}|^2
        \Bigr)^{-1}
    \right].
\end{equation} where $c,d,e,C,D,E$ are defined through the coupled systems
\eqref{system1} and \eqref{system2}.

These equations are valid for sufficiently large $\omega$. We then perform
an analytic continuation to $\omega \to 0$, taking care that the continuation
may involve a change of branch. As recalled in~\cite{bousseyroux1},
by taking $\omega = -i\epsilon$ with $\epsilon>0$ and $\epsilon \to 0$, we have
\begin{equation}
       \g_{1,\vb{M}}(\omega,z) \longrightarrow i\,\oo_{\vb{M}}(z), 
\end{equation} where $\oo_{\vb{M}}(z)$, defined in~\cite{bousseyroux1},
measures the non-normality of the eigenvectors. 

In the limit $\omega \to 0$, the previous equations become
\begin{equation}
    i\,\oo(z)
    =
    -d\,\tau\!\left[
        \Bigl(
            cd\,\vb{A}\vb{A}^*
            -
            |z-\overline{e}\vb{A}|^2
        \Bigr)^{-1}
    \right].
\end{equation}
\begin{equation}
    i\,\oo(z)
    =-c
    \tau\!\left[
        \vb{A}\vb{A}^*
        \Bigl(
            cd\,\vb{A}\vb{A}^*
            -
            |z-\overline{e}\vb{A}|^2
        \Bigr)^{-1}
    \right].
\end{equation}
\begin{equation}
    \g(z)
    =
    -\tau\!\left[
        (z-\overline{e}\vb{A})^*
        \Bigl(
            cd\,\vb{A}\vb{A}^*
            -
            |z-\overline{e}\vb{A}|^2
        \Bigr)^{-1}
    \right].
\end{equation}
where $c,d,e,C,D,E$ are defined through the coupled systems \eqref{system1} and

\begin{equation}
    \left\{
    \begin{aligned}
        D
        &=
        \tau\!\left[
            -c\vb{A}\vb{A}^*
            \Bigl(
                cd\,\vb{A}\vb{A}^*
                -
                |z-\overline{e}\vb{A}|^2
            \Bigr)^{-1}
        \right],
        \\[0.6em]
        E
        &=
       - \tau\!\left[
            \vb{A}(z-\overline{e}\vb{A})^*
            \Bigl(
                cd\,\vb{A}\vb{A}^*
                -
                |z-\overline{e}\vb{A}|^2
            \Bigr)^{-1}
        \right],
        \\[0.6em]
        C
        &=
        \tau\!\left[
            -d\,\vb{A}\vb{A}^*
            \Bigl(
                cd\,\vb{A}\vb{A}^*
                -
                |z-\overline{e}\vb{A}|^2
            \Bigr)^{-1}
        \right].
    \end{aligned}
    \right.
\end{equation}
We introduce the variables $\gamma = \sqrt{cd}$ and $G = \sqrt{CD}$, and then obtain the main result \ref{mainresult}.

\subsection{Derivation of the result \ref{result_boundaries}}
\label{second_result}

Let $z'$ be a point on the spectral boundary of $\vb{A}\vb{B}$. The idea is again to study the limit of the previous equations as $z\to z'$, knowing that
\begin{equation}
    \mathfrak{o}(z)\longrightarrow 0
\end{equation}
in this limit. Since the left-hand side of the first equation of the system ~\eqref{system_zero} tends to zero, we must have
\begin{equation}
    G\longrightarrow 0.
\end{equation}
There are then two possible regimes:
\begin{equation}
    \mathcal{R}_{1,\vb{B}}(G,E)\longrightarrow 0
    \qquad
    \text{or}
    \qquad
    \mathcal{R}_{1, \vb{B}}(G,E)\longrightarrow +\infty.
\end{equation}
We first consider the case
\begin{equation}
    \mathcal{R}_{1, \vb{B}}(G,E)\longrightarrow 0.
\end{equation}
Taking the limit $G\to0$ in~\eqref{system_zero}, and keeping the first non-vanishing order in the first equation, we obtain
\begin{equation}\label{boundary_system_zero}
    \left\{
    \begin{aligned}
        1
        &=
        \tau\left[
            \partial_{\alpha}\mathcal{R}_{1, \vb{B}}(0,E)\,
            \vb{A}^*\vb{A}
            \left(
                \left|z-\mathcal{R}_{2, \vb{B}}(0,E)\vb{A}\right|^2
            \right)^{-1}
        \right],
        \\[0.6em]
        E
        &=
        \tau\left[
            \vb{A}
            \left(
                z-\mathcal{R}_{2, \vb{B}}(0,E)\vb{A}
            \right)^{-1}
        \right].
    \end{aligned}
    \right.
\end{equation}

We now set
\begin{equation}
    x
    =
    \frac{z}{\mathcal{R}_{2,\vb{B}}(0,E)}.
\end{equation}
Then~\eqref{boundary_system_zero} becomes
\begin{equation}
    \left\{
    \begin{aligned}
        \frac{1}{
            \tau\left[
                \vb{A}^*\vb{A}
                \left(
                    |x\vb{1}-\vb{A}|^2
                \right)^{-1}
            \right]
        }
        &=
        \frac{
            \partial_{\alpha}\mathcal{R}_{1, \vb{B}}(0,E)
        }{
            \left|\mathcal{R}_{2, \vb{B}}(0,E)\right|^2
        },
        \\[0.8em]
        E\mathcal{R}_{2, \vb{B}}(0,E)
        &=
        \mathfrak{t}_{\vb{A}}(x).
    \end{aligned}
    \right.
\end{equation}
Using the definition of the associated $S$-transform, this gives
\begin{equation}
    \mathcal{R}_{2, \vb{B}}(0,E)
    =
    \frac{1}{
        S_{\vb{B}}\left(\mathfrak{t}_{\vb{A}}(x)\right)
    }.
\end{equation} which gives the result.

\paragraph{Second case.}

Consider now the case where $|\mathcal{R}_{1, \vb{B}}(0,E)| \to +\infty$. From the second equation of \eqref{system_zero}, this forces us to get $E\to 0$.
One then realizes that what ultimately matters are only the quantities $\tau(\vb{A}^{-1})$ and $\tau([\vb{A}\vb{A}^*]^{-1})$. Thus, to determine the spectral boundary, it suffices to replace $\vb{A}^{-1}$ by, for instance, a matrix of the form $\tau(\vb{A}^{-1}) + \sqrt{\tau([\vb{A}\vb{A}^*]^{-1}) - |\tau(\vb{A}^{-1})|^2}\,\vb{G}$ where $\vb{G}$ is a Ginibre matrix defined in Definition \ref{def:ginibre}. We can then use the type $1$ equations given by \eqref{biinvariantcase}, by considering the inverse matrix.

\subsection{Proof of Lemma \ref{propinverse}}

Let us finally prove the Lemma \ref{propinverse}. We first recall the notation. For a large random matrix $\vb{A}$, we set
\begin{equation}
    \g_{\vb A}(z)
    =
    \tau\bigl((z\vb 1-\vb A)^{-1}\bigr),
\end{equation}
and
\begin{equation}
    \mathfrak t_{\vb A}(z)
    =
    \tau\bigl(\vb A(z\vb 1-\vb A)^{-1}\bigr).
\end{equation}
We also introduce the two functions
\begin{equation}
    h_{\vb A}(z)
    =
    \tau\left(|z\vb 1-\vb A|^{-2}\right),
\end{equation}
and
\begin{equation}
    f_{\vb A}(z)
    =
    \tau\left(
    \bigl[(z\vb 1-\vb A)^{-1}\vb A\bigr]^*
    \bigl[(z\vb 1-\vb A)^{-1}\vb A\bigr]
    \right).
\end{equation}
Finally, we shall use the relation, proved in~\cite{bousseyroux1},
\begin{equation}\label{relation_bousseyroux1}
    \partial_\alpha
    \mathcal R_{1,\vb A}(0,\g_{\vb A}(z))
    =
    \frac{1}{|\g_{\vb A}(z)|^2}
    -
    \frac{1}{h_{\vb A}(z)}.
\end{equation}

\begin{lemma}
    Let $\vb A$ be a large random matrix. The expression
    \begin{equation}
        \partial_\alpha
        \mathcal R_{1,\vb A^{-1}}(0,g)
        =
        \frac{
            |\mathcal R_{2,\vb A^{-1}}(0,g)|^4
            \,
            \partial_\alpha
            \mathcal R_{1,\vb A}(0,u)
        }
        {
            1
            -
            |\mathcal R_{2,\vb A^{-1}}(0,g)|^2
            \left(
                1
                +
                2\Re\bigl(g\mathcal R_{2,\vb A^{-1}}(0,g)\bigr)
            \right)
            \partial_\alpha
            \mathcal R_{1,\vb A}(0,u)
        },
    \end{equation}
    where
    \begin{equation}
        u
        =
        -\mathcal R_{2,\vb A^{-1}}(0,g)
        -
        g\mathcal R_{2,\vb A^{-1}}(0,g)^2,
    \end{equation}
    can provide one branch of
    \begin{equation}
        g\mapsto
        \partial_\alpha\mathcal R_{1,\vb A^{-1}}(0,g),
    \end{equation}
    where $\mathcal R_{2,\vb A^{-1}}$ and $\mathcal R_{1,\vb A}$ denote suitable branches associated respectively with $\vb A^{-1}$ and $\vb A$.
\end{lemma}

\begin{proof}
    We start from the elementary identity
    \begin{equation}
        (z\vb 1-\vb A)^{-1}\vb A
        =
        z(z\vb 1-\vb A)^{-1}-\vb 1.
    \end{equation}
    It gives
    \begin{equation}
        f_{\vb A}(z)
        =
        |z|^2h_{\vb A}(z)
        -
        2\Re\bigl(z\g_{\vb A}(z)\bigr)
        +
        1.
    \end{equation}
    Applying this identity to $1/z$, we obtain
    \begin{equation}
        h_{\vb A^{-1}}(z)
        =
        \tau\left(|z\vb 1-\vb A^{-1}|^{-2}\right)
        =
        \frac{1}{|z|^2}f_{\vb A}(1/z).
    \end{equation}
    Hence
    \begin{equation}
        h_{\vb A^{-1}}(z)
        =
        \frac{1}{|z|^2}
        \left(
            \frac{1}{|z|^2}h_{\vb A}(1/z)
            -
            2\Re\left(\frac{1}{z}\g_{\vb A}(1/z)\right)
            +
            1
        \right).
    \end{equation}

    By~\eqref{relation_bousseyroux1}, we have equivalently
    \begin{equation}
        h_{\vb A}(z)
        =
        \frac{
            |\g_{\vb A}(z)|^2
        }
        {
            1
            -
            |\g_{\vb A}(z)|^2
            \partial_\alpha\mathcal R_{1,\vb A}(0,\g_{\vb A}(z))
        }.
    \end{equation}

    We also have
    \begin{equation}
        \g_{\vb A}(1/z)
        =
        -z\,\mathfrak t_{\vb A^{-1}}(z).
    \end{equation}
    If we write
    \begin{equation}
        t=\mathfrak t_{\vb A^{-1}}(z),
    \end{equation}
    then
    \begin{equation}
        \g_{\vb A}(1/z)=-zt.
    \end{equation}
    Therefore
    \begin{equation}
        h_{\vb A}(1/z)
        =
        \frac{
            |z|^2|t|^2
        }
        {
            1
            -
            |z|^2|t|^2
            \partial_\alpha
            \mathcal R_{1,\vb A}(0,-zt)
        }.
    \end{equation}
    Substituting this expression into the formula for $h_{\vb A^{-1}}(z)$ gives
    \begin{equation}
        h_{\vb A^{-1}}(z)
        =
        \frac{1}{|z|^2}
        \left(
            \frac{
                |t|^2
            }
            {
                1
                -
                |z|^2|t|^2
                \partial_\alpha
                \mathcal R_{1,\vb A}(0,-zt)
            }
            +
            2\Re(t)
            +
            1
        \right).
    \end{equation}

    We now set
    \begin{equation}
        g=\g_{\vb A^{-1}}(z),
        \qquad
        t=\mathfrak t_{\vb A^{-1}}(z).
    \end{equation}
    Since
    \begin{equation}
        t=zg-1,
    \end{equation}
    and
    \begin{equation}
        z=\frac{1}{g}+\mathcal R_{2,\vb A^{-1}}(0,g),
    \end{equation}
    we obtain
    \begin{equation}
        t
        =
        g\mathcal R_{2,\vb A^{-1}}(0,g).
    \end{equation}
    Thus
    \begin{equation}
        zt
        =
        \left(
            \frac{1}{g}
            +
            \mathcal R_{2,\vb A^{-1}}(0,g)
        \right)
        g\mathcal R_{2,\vb A^{-1}}(0,g),
    \end{equation}
    that is,
    \begin{equation}
        zt
        =
        \mathcal R_{2,\vb A^{-1}}(0,g)
        +
        g\mathcal R_{2,\vb A^{-1}}(0,g)^2.
    \end{equation}
    Consequently,
    \begin{equation}
        -zt
        =
        -\mathcal R_{2,\vb A^{-1}}(0,g)
        -
        g\mathcal R_{2,\vb A^{-1}}(0,g)^2.
    \end{equation}

    To simplify the notation, write
    \begin{equation}
        R=\mathcal R_{2,\vb A^{-1}}(0,g),
    \end{equation}
    and
    \begin{equation}
        D_A(g)
        =
        \partial_\alpha
        \mathcal R_{1,\vb A}
        \left(
            0,
            -R-gR^2
        \right).
    \end{equation}
    Since
    \begin{equation}
        z=\frac{1+gR}{g},
    \end{equation}
    we have
    \begin{equation}
        \frac{1}{|z|^2}
        =
        \frac{|g|^2}{|1+gR|^2}.
    \end{equation}
    Moreover,
    \begin{equation}
        R+gR^2=R(1+gR),
    \end{equation}
    and therefore
    \begin{equation}
        |R+gR^2|^2
        =
        |R|^2|1+gR|^2.
    \end{equation}
    Hence
    \begin{equation}
        h_{\vb A^{-1}}(z)
        =
        \frac{|g|^2}{|1+gR|^2}
        \left(
            1
            +
            2\Re(gR)
            +
            \frac{|gR|^2}
            {
                1
                -
                |R|^2|1+gR|^2D_A(g)
            }
        \right).
    \end{equation}

    Using~\eqref{relation_bousseyroux1} for $\vb A^{-1}$, we get
    \begin{equation}
        \partial_\alpha
        \mathcal R_{1,\vb A^{-1}}(0,g)
        =
        \frac{1}{|g|^2}
        -
        \frac{1}{h_{\vb A^{-1}}(z)}.
    \end{equation}
    Since
    \begin{equation}
        |1+gR|^2
        =
        1
        +
        2\Re(gR)
        +
        |gR|^2,
    \end{equation}
    a direct simplification gives
    \begin{equation}
        \partial_\alpha
        \mathcal R_{1,\vb A^{-1}}(0,g)
        =
        \frac{
            |R|^4D_A(g)
        }
        {
            1
            -
            |R|^2
            \left(
                1+2\Re(gR)
            \right)
            D_A(g)
        }.
    \end{equation}
    Replacing $R$ and $D_A(g)$ by their definitions gives exactly
    \begin{equation}
        \partial_\alpha
        \mathcal R_{1,\vb A^{-1}}(0,g)
        =
        \frac{
            |\mathcal R_{2,\vb A^{-1}}(0,g)|^4
            \,
            \partial_\alpha
            \mathcal R_{1,\vb A}(0,u)
        }
        {
            1
            -
            |\mathcal R_{2,\vb A^{-1}}(0,g)|^2
            \left(
                1
                +
                2\Re\bigl(g\mathcal R_{2,\vb A^{-1}}(0,g)\bigr)
            \right)
            \partial_\alpha
            \mathcal R_{1,\vb A}(0,u)
        },
    \end{equation}
    with
    \begin{equation}
        u
        =
        -\mathcal R_{2,\vb A^{-1}}(0,g)
        -
        g\mathcal R_{2,\vb A^{-1}}(0,g)^2.
    \end{equation}
    This concludes the proof.
\end{proof}

\subsection{Proof of Eq.\eqref{formula_produit} from the main result \ref{mainresult}}\label{from}

Let $\vb{M}$ be a bi-invariant random matrix. Using Gauss's law and the radial symmetry of the spectral distribution of $\vb{M}$, one obtains
\begin{equation}\label{g}
    \mathfrak{g}_{\vb{M}}(z) = \frac{F_{\vb{M}}(|z|)}{z}.
\end{equation}
Moreover, from~\cite{belinschi2017squared}, we have
\begin{equation}\label{o}
    \mathfrak{o}_{\vb{M}}(z)^2
    =
    \frac{F_{\vb{M}}(|z|)\bigl(1-F_{\vb{M}}(|z|)\bigr)}{|z|^2},
\end{equation}
where $F_{\vb{M}}$ is defined in Eq.~\eqref{defF}.

From the introduction of~\cite{bousseyroux1}, one obtains that, when $\vb{B}$ is a bi-invariant matrix,
\begin{equation}
    \mathcal{R}_{1,\vb{B}}(\alpha) = R_{\mathbb{B}}(\alpha),
\end{equation}
where
\begin{equation}\label{hermitianise}
\mathbb{B} =
\begin{pmatrix}
0 & \vb{B}\\
\vb{B}^\ast & 0
\end{pmatrix},
\end{equation}
and
\begin{equation}
    \mathcal{R}_{2,\vb{B}} = 0.
\end{equation}
Furthermore,
\begin{equation}
    \mathfrak{g}_{\mathbb{B}}(z)
    =
    z\,\mathfrak{g}_{\vb{B}\vb{B}^\ast}(z^2),
\end{equation}
and hence, using \eqref{definitionS_classical}, we have
\begin{equation}\label{rela}
    S_{\vb{B}\vb{B}^\ast}(t)
    =
    \frac{t}{t+1}\,S_{\mathbb{B}}(t)^2.
\end{equation}

We are now ready to prove \eqref{formula_produit}. Let $\vb{A}$ and $\vb{B}$ be two bi-invariant random matrices. Using the main result~\ref{mainresult}, we have
\begin{equation}
    \mathfrak{g}_{\vb{A}\vb{B}}(z)
    =
    -\tau\left[
        z^\ast
        \left(
            \mathcal{R}_{1,\vb{B}}(G)^2\vb{A}\vb{A}^\ast
            -
            |z|^2
        \right)^{-1}
    \right],
\end{equation}
\begin{equation}
    \mathfrak{o}_{\vb{A}\vb{B}}(z)^2
    =
    G\mathcal{R}_{1,\vb{B}}(G)
    \,
    \tau\left[
        \left(
            \mathcal{R}_{1,\vb{B}}(G)^2\vb{A}\vb{A}^\ast
            -|z|^2
        \right)^{-1}
    \right],
\end{equation}
and
\begin{equation}\label{eq30}
    G
    =
    -\mathcal{R}_{1,\vb{B}}(G)
    \,
    \tau\left[
        \vb{A}\vb{A}^\ast
        \left(
            \mathcal{R}_{1,\vb{B}}(G)^2\vb{A}\vb{A}^\ast
            -|z|^2
        \right)^{-1}
    \right].
\end{equation}
The first two equations give
\begin{equation}
    -\mathfrak{g}_{\vb{A}\vb{B}}(z)\,
    G\,\mathcal{R}_{1,\vb{B}}(G)
    =
    z^\ast \mathfrak{o}_{\vb{A}\vb{B}}(z)^2.
\end{equation}
Equations~\eqref{g} and~\eqref{o} then imply that
\begin{equation}\label{eq4}
    G\,\mathcal{R}_{1,\vb{B}}(G)
    =
    F_{\vb{A}\vb{B}}(|z|)-1.
\end{equation}
Therefore, using \eqref{rela}, we obtain
\begin{equation}\label{eq45}
    S_{\vb{B}\vb{B}^\ast}\bigl(F_{\vb{A}\vb{B}}(|z|)-1\bigr)
    =
    \frac{F_{\vb{A}\vb{B}}(|z|)-1}
    {F_{\vb{A}\vb{B}}(|z|)\mathcal{R}_{1,\vb{B}}(G)^2}.
\end{equation}
Going back to equation~\eqref{eq30}, we get
\begin{equation}
    G\mathcal{R}_{1,\vb{B}}(G)
    =
    \mathfrak{t}_{\vb{A}\vb{A}^\ast}
    \left(
        \left\lvert
        \frac{z}{\mathcal{R}_{1,\vb{B}}(G)}
        \right\rvert^2
    \right),
\end{equation}
and therefore
\begin{equation}
    F_{\vb{A}\vb{B}}(|z|)-1
    =
    \mathfrak{t}_{\vb{A}\vb{A}^\ast}
    \left(
        \left\lvert
        \frac{z}{\mathcal{R}_{1,\vb{B}}(G)}
        \right\rvert^2
    \right).
\end{equation}
The relation \eqref{definitionS_classical} gives
\begin{equation}
    S_{\vb{A}\vb{A}^\ast}\bigl(F_{\vb{A}\vb{B}}(|z|)-1\bigr)
    =
    \frac{F_{\vb{A}\vb{B}}(|z|)}
    {F_{\vb{A}\vb{B}}(|z|)-1}
    \frac{|\mathcal{R}_{1,\vb{B}}(G)|^2}
    {|z|^2}.
\end{equation}
Thus, with equation~\eqref{eq45}, we obtain
\begin{equation}
    S_{\vb{A}\vb{A}^\ast}\bigl(F_{\vb{A}\vb{B}}(|z|)-1\bigr)
    S_{\vb{B}\vb{B}^\ast}\bigl(F_{\vb{A}\vb{B}}(|z|)-1\bigr)
    =
    \frac{1}{|z|^2}.
\end{equation}

\end{document}